\documentclass[11pt]{article}

\newcommand{\ifabs}[2]{#2}

\usepackage{comment}
\usepackage{cite}
\usepackage{fullpage}
\usepackage{color}
\usepackage{epic}
\usepackage{eepic}
\usepackage{epsfig}
\usepackage{xspace}
\usepackage{graphicx}

\usepackage{amsmath}
\usepackage{amsfonts}
\usepackage{amssymb}
\usepackage{tikz}
\usepackage{mathrsfs}
\usepackage{multirow}
\usepackage{calc}
\usepackage{algorithm}
\usepackage{algorithmicx}
\usepackage{algpseudocode}

\ifabs{ \excludecomment{fullproof}

\includecomment{abs}
\excludecomment{fullpaper}

}{

\excludecomment{proofsketch} \excludecomment{abs}
\includecomment{fullpaper}
}

\excludecomment{suppress}

\newcommand{\todo}[1]{\typeout{TODO: \the\inputlineno: #1}\textbf{[[[ #1 ]]]}}

\newcommand{\concept}[1]{\emph{#1}}

\newtheorem{theorem}{Theorem}
\newtheorem{lemma}[theorem]{Lemma}
\newtheorem{corollary}[theorem]{Corollary}
\newtheorem{definition}[theorem]{Definition}
\newtheorem{definitions}[theorem]{Definitions}

\newtheorem{claim}[theorem]{Claim}

\newtheorem{proposition}[theorem]{Proposition}
\newtheorem{assumption}{Assumption}

\newcommand{\newloglike}[2]{\newcommand{#1}{\mathop{\rm #2}\nolimits}}
\newloglike{\sgn}{sgn}

\newcommand{\nul}[1]{{\it et al.\/}}

\newenvironment{proof}{\noindent{\bf Proof: }}{\nopagebreak\rule{1 ex}{0.8 em}\medskip}

\newenvironment{remark}{\noindent{\bf Remark: }}{\nopagebreak\medskip}

\newcommand{\lANN}{\lambda\text{-}\mathsf{ANN}^{\gamma}_{d,n}}
\newcommand{\lANNS}{\lambda\text{-}\mathsf{ANNS}^{\gamma}_{d,n}}
\newcommand{\ANNS}{\mathsf{ANNS}^{\gamma}_{d,n}}
\newcommand{\LPM}{\mathsf{LPM}^\Sigma_{m,n}}

\begin{document}
\title{Randomized approximate nearest neighbor search\\ with limited adaptivity}
\author{
	Mingmou Liu~\thanks{Department of Computer Science and Technology, Nanjing University, China. \texttt{MG1533028@smail.nju.edu.cn}}
	\and
	Xiaoyin Pan~\thanks{Department of Computer Science and Technology, Nanjing University, China. \texttt{xiaoyin.pan95@gmail.com}}
	\and Yitong Yin~\thanks{State Key Laboratory for Novel Software Technology, Nanjing University, China. Supported by NSFC grants no.~61272081 and 61321491. Email: \texttt{yinyt@nju.edu.cn}.}  }

\date{}

\maketitle

\begin{abstract}
We study the fundamental problem of approximate nearest neighbor search in $d$-dimensional Hamming space $\{0,1\}^d$. 
We study the complexity of the problem in the famous cell-probe model, a classic model for data structures. 
We consider algorithms in the cell-probe model with \emph{limited adaptivity},
where the algorithm makes $k$ rounds of parallel accesses to the data structure for a given $k$.
For any $k\ge 1$, we give a simple randomized algorithm solving the approximate nearest neighbor search using $k$ rounds of parallel memory accesses, with $O(k(\log d)^{1/k})$ accesses in total.
We also give a more sophisticated randomized algorithm using $O(k+(\frac{1}{k}\log d)^{O(1/k)})$ memory accesses in $k$ rounds for large enough $k$. Both algorithms use data structures of size polynomial in $n$, the number of points in the database.

For the lower bound, we prove an $\Omega(\frac{1}{k}(\log d)^{1/k})$ lower bound for the total number of memory accesses required by any randomized algorithm solving the approximate nearest neighbor search within $k\le\frac{\log\log d}{2\log\log\log d}$ rounds of parallel memory accesses on any data structures of polynomial size.
This lower bound shows that our first algorithm is asymptotically optimal for any constant round $k$. And our second algorithm approaches the asymptotically optimal tradeoff between rounds and memory accesses, in a sense that the lower bound of memory accesses for any $k_1$ rounds can be matched by the algorithm within $k_2=O(k_1)$ rounds. 
In the extreme, for some large enough $k=\Theta\left(\frac{\log\log d}{\log\log\log d}\right)$, our second algorithm matches the $\Theta\left(\frac{\log\log d}{\log\log\log d}\right)$ tight bound for fully adaptive algorithms for approximate nearest neighbor search due to Chakrabarti and Regev~\cite{CharkReg}.

\end{abstract}

\section{Introduction}

Nearest neighbor search is a fundamental theoretical problem in Computer Science, with enormously many applications in diverse fields. In the nearest neighbor search problem, we are given a database $B$ of $n$ points from a metric space $X$. The goal is to preprocess them into a data structure, such that given any query point $x\in X$, an algorithm with accessing to the data structure can find a database point in $B$ that is closest to the query point $x$ among all database points. An extensively studied case is when the metric space is the Hamming space $X=\{0,1\}^d$.

It is conjectured that the nearest neighbor search is hard to solve by any data structures when the dimension $d$ is high (e.g.~$d\gg\log n$). This conjecture is sometimes referred as a case of the ``curse of dimensionality'' and is one of the central problems in the area of data structure lower bounds. It is also believed that the problem of high-dimensional nearest neighbor search remains to be intractable while either an approximation is tolerated or the algorithm is randomized, but not both at the same time~\cite{indyk2004nearest}. 

The complexity of the nearest neighbor search problem, as well as many other data structure problems, is well studied in the cell-probe model~\cite{yao1981should}, a classic model for the complexity of data structures. In the cell-probe model, the database is preprocessed into a data structure,  stored as a table in the main memory, and upon each query, an algorithm, called the cell-probing algorithm, outputs an answer to the query after adaptively probing a number of table cells. The complexity is measured by both the size of the data structure and the number of cell-probes made by the algorithm to answer a query in the worst case. There is a substantial body of works on the cell-probe complexity of nearest neighbor search in Hamming space~\cite{borodin1999lower,BarkRaba,jayram2003cell,liu2004strong, Mihai06, PTW08,panigrahy2010lower}.

When both {approximation} and {randomization} are allowed, a seminal work of Chakrabarti and Regev~\cite{CharkReg} gives a \emph{tight} bound for the complexity of  nearest neighbor search in $d$-dimensional Hamming space with data structures of size polynomial in $n$, assuming the dimension $d$ is high (and not too high to trivialize the problem, e.g.~$(\log n)^{1.01} \le d \le 2^{\sqrt{\log n}}$). This fundamental result is stated informally as follows.

\begin{theorem}[Chakrabarti and Regev~\cite{CharkReg}]\label{theorem-CharkReg}
Assume $(\log n)^{1.01} \le d \le 2^{\sqrt{\log n}}$. The cell-probe complexity of randomized approximate nearest neighbor search in $d$-dimensional Hamming space in the cell-probe model with data structure of polynomial size is $\Theta\left(\frac{\log\log d}{\log\log\log d}\right)$.
\end{theorem}

On the other hand, when the table size becomes closer to be linear of $n$, data structures such as locality-sensitive hashing (LSH)~\cite{IM98,andoni2006near} or data-dependent LSH~\cite{andoni2014beyond,andoni2015optimal} achieve a cell-probe complexity of $\tilde{O}(dn^{\rho})$ with data structures of size $\tilde{O}(n^{1+\rho})$ for some $0<\rho<1$ depending on the metric and the approximation ratio. Compared to the $\Theta\left(\frac{\log\log d}{\log\log\log d}\right)$ bound of Chakrabarti and Regev, the $\tilde{O}(dn^{\rho})$ cell-probe complexity is much worse. The popularity in practice of the LSH-based data structures is due to their low space cost, and the ability to be implemented in parallel.

Take locality-sensitive hashing (LSH) as an example. The algorithm of LSH is \emph{non-adaptive}: Each cell-probe relies only on the query but not on the information retrieved by other cell-probes. This makes all cell-probes in LSH parallelizable into one round of parallel memory accesses. And the more recent data-dependent LSH~\cite{andoni2014beyond,andoni2015optimal} surpasses the classic LSH in cell-probe complexity by being a little more adaptive: the algorithm retrieves a data-dependent hash function before making the second round of cell-probes, while the cell-probes in the second round are independent of each other.
In contrast, the algorithm of Chakrabarti and Regev~\cite{CharkReg} is \emph{fully} adaptive: Every cell-probe must wait for the information retrieved by the previous cell-probe to proceed.

This could give us the following intuitive image:
A cell-probing algorithm is getting more clever and hence more efficient,  as it is becoming more adaptive.
It is then a fundamental question to study the tradeoff between the efficiency (measured by the total number of cell-probes) and adaptivity (measured by the number of rounds of parallel cell-probes) in the cell-probe model. Very little was known to this fundamental question. In~\cite{brody2015adapt}, Brody and Larsen initiated the study of non-adaptive \emph{dynamic} data structures, where the database receives both queries and updates. They show a foundational result that for dynamic data structures, being adaptive is critical. For static data structures, parallel cell-probes were studied in the context of low-contention data structures~\cite{aspnes2010low, panigrahy2010lower}. 
The highest cell-probe lower bound known for \emph{non-adaptive} static data structure is the $\Omega(\log n/\log\frac{sd}{n\log n})$ cell-probe lower bound of Panigrahy, Talwar, and Wieder~\cite{PTW08} for randomized approximate nearest neighbor search on a table of size $s$. This lower bound becomes trivial for tables of polynomial size. For cell-probe model with polynomial-sized data structures, the tradeoff between the cell-probe complexity and adaptivity is highly unknown for any static data structure problems.

\begin{suppress}
\todo{TODO: related works}
\par
Nearest Neighbor Search problem is a fundamental theoretical problem in computer science, with various applications in very diverse fields, such as molecular biology, information retrieval, pattern recognition and machine learning. In most of the filed, the interesting objects is abstacted to points in some metric space, then the algorithms try to find the objects which are similar with another interesting object in the new metric space.
However, it is widely believed that the nearest neighbor search problem suffers from the ``curse of dimensionality'': without randomization or approximation, the cost of computing the nearest neighbor queryies may increase drastically as the dimension of the underlying metric space grows.
In fact, nearest neighbor search is proved that it is solvable with both randomization and approximation.
Within polynomial space,~\cite{CharkReg} show a $\Theta(\frac{\log\log d}{\log\log\log d})$ tight bound.
Within much more practical space, there are upper bounds including Johnson-Lindenstrauss, Locality Sentivity Hashing, and data-dependent Locality Sentivity Hashing. And a $\Omega(\log n/\log\frac{sw}{n})$ lower bound proved Panigraphy \emph{et al.}~\cite{PTW08}~\cite{panigrahy2010lower}.
However, all the practical solutions use very limited adaptiveness. In another word, all the solutions access the data structure in very limited rounds, even only $1$-round or $2$-round.
For lower bound, very few is known. Only~\cite{PTW08} has a lower bound for non-adaptive (i.e. the algorithm accesses the data structure in only $1$-round) approximate nearest neighbor search, $\Omega(\log n/\log(\frac{sw}{n\log n}))$, which becomes trivial in polynomial space.
We are curious. What is the power of adaptiveness in the fundamental problem of approximate nearest neighbor search? What is the exact tradeoff between the number of rounds of data structure accessing and the total number of data structure accessing?
With out lost of generality, we suppose algorithms access the data structures in $k$ rounds. We give two upper bounds: a simple algorithm with cell-probe complexity of $O(k(\log d)^{1/k})$ and a more sophisticated algorithm with cell-probe complexity of $O(k+(\frac{1}{k}\log d)^{c/k})$ for any constant $c>2$, and a lower bound which says that every algorithm has cell-probe complexity of $\Omega((\log_\gamma d)^{1/k}/k)$.
We emphasize that the first algorithm has the optimal cell-probe complexity when $k$ is constant.
On the other hand, the second algorithm has the optimal round complexity, approach the asymptotically optimal tradeoff between rounds and cell-probe complexity: any cell-probe complexity lower bound in $k$ round can be achieved by the algorithm in $O(k)$ round.
\par
\end{suppress}

\ifabs{
\paragraph{Our results}
}{
\paragraph{Our results.}}
In this paper, we study the complexity of randomized approximate nearest neighbor search in the cell-probe model with limited adaptivity. We consider a natural notion of \concept{rounds} for cell-probes, where the cell-probes in the same round are adaptive to the information retrieved in previous rounds, but non-adaptive to each other in the same round.

We give two randomized cell-probing algorithms for approximate nearest neighbor search in $d$-dimensional Hamming space. For both algorithms, the data structures are of polynomial size, and the cell-probes are organized into $k$ rounds for any $k\ge 1$ (Algorithm~\ref{alg:simple}) or for all sufficiently large $k$ (Algorithm~\ref{alg:general}). 
The first algorithm is as follow.

\begin{theorem}\label{theorem-simple-upper-bound-informal}
For any $k\ge 1$, the approximate nearest neighbor search in $d$-dimensional Hamming space can be solved in the cell-probe model with a data structure of polynomial size,
using $k$ rounds of parallel randomized cell-probes, with 
$O\left(k(\log d)^{1/k}\right)$ cell-probes in total.
\end{theorem}

The algorithm is simple and works for all $k\ge 1$ number of rounds. Especially when $k=1$, the algorithm is non-adaptive. Compared to the LSH which is also non-adaptive, our algorithm achieves a much better cell-probe complexity $O(\log d)$ by using a data structure of larger polynomial size.

However, when the round number $k$ becomes large, especially at the extreme when every round has 1 cell-probe, in which case the algorithm becomes fully adaptive and has $O(\log\log d)$ total cell-probes, which is not optimal for fully adaptive algorithms by Theorem~\ref{theorem-CharkReg}. This leads us to our second more sophisticated algorithm. 

\begin{theorem}\label{theorem-general-upper-bound-informal}
For large enough $k$, the approximate nearest neighbor search in $d$-dimensional Hamming space can be solved in the cell-probe model with a data structure of polynomial size, using $k$ rounds of parallel randomized cell-probes, with 
$O\left(k+\left(\frac{1}{k}\log d\right)^{O(1/k)}\right)$ cell-probes in total.
\end{theorem}

The second algorithm is substantially more sophisticated. In the extreme, it approaches the optimal fully adaptive algorithm in Theorem~\ref{theorem-CharkReg} in the following sense:
For some sufficiently large $k=O\left(\frac{\log\log d}{\log\log\log d}\right)$, we can implement the algorithm such that every round of the algorithm contain only 1 cell-probe.


We emphasize that these algorithms are not meant to be efficient in practice due to their expensive space costs, rather, they are parts of a theoretical endeavor to understand the complexity tradeoff between time and rounds on data structures of polynomial size.
With this spirit, we prove the following lower bound for the tradeoff between cell-probe complexity and round complexity for  randomized approximate nearest neighbor search. 
\begin{theorem}\label{theorem-ANNS-lower-bound-informal}
Assume $(\log n)^{1.01}\leq d \leq 2^{\sqrt{\log n}}$ and $1\le k\le\frac{\log\log d}{2\log\log\log d}$.
Any randomized algorithm solving the approximate nearest neighbor search in $d$-dimensional Hamming space in the cell-probe model with a data structure of polynomial size using $k$ rounds of parallel randomized cell-probes must use $\Omega\left(\frac{1}{k}(\log d)^{1/k}\right)$ cell-probes in total.
\end{theorem}
Due to this lower bound, both our algorithms achieve some optimality:
\begin{itemize}
\item Algorithm~\ref{alg:simple} is asymptotically optimal in cell-probe complexity for any constant number of rounds.
\item Algorithm~\ref{alg:general} approaches the asymptotically optimal tradeoff between cell-probe complexity and round complexity in the following sense: the cell-probe lower bound for any $k_1$-round algorithms can be approached by Algorithm~\ref{alg:general} within $k_2=O(k_1)$  rounds.
\end{itemize}
In addition, Algorithm~\ref{alg:general} together with our lower bound show that the cell-probe complexity of randomized approximate nearest neighbor search undergoes a ``phase transition'' when the round number is within the regime $k=\Theta\left(\frac{\log\log d}{\log\log\log d}\right)$: For a small $k_1=\Theta\left(\frac{\log\log d}{\log\log\log d}\right)$, the average number of cell-probes per each round for any $k_1$-round algorithm has to be a $(\log\log d)^{\Omega(1)}$, whereas for large enough $k_2=\Theta\left(\frac{\log\log d}{\log\log\log d}\right)$,  only 1 cell-probe in each round is enough for a $k_2$-round algorithm.

\ifabs{\paragraph{Technique}}{\paragraph{Technique.} }
Both our upper bounds and lower bounds rely heavily on the machineries developed in~\cite{CharkReg}.

The main ideas for the upper bounds are the dimension reduction techniques developed in the pioneering works of~\cite{kushilevitz2000efficient, IM98} and the multi-way search in~\cite{CharkReg}. Our efforts are focused on how to apply these techniques to give a \emph{family} of algorithms approaching the smoothed tradeoff between round and cell-probe complexity. 
A technical innovation of~\cite{CharkReg} is to use two kinds of approximations of Hamming balls: an accurate approximation of hamming ball which is more expensive, and a coarse approximation which is cheap, to support a multi-way search with a substantial number of branchings, such that each branching is supported by one query to an accurate ball succeeded by several queries to coarse balls, which altogether consume only $O(1)$ cell-probes. 
Surprisingly, we discover that a simple algorithm can achieve an optimal cell-probe complexity in any constant number of rounds, using only the more expensive accurate approximation of Hamming balls.
And for general round numbers, the coarse approximation of balls are employed to approach the asymptotically optimal tradeoff between rounds and cell-probes.

The lower bound is proved by the round elimination of communication protocols for the longest prefix matching problem $\mathsf{LPM}$, which can be reduced to approximate nearest neighbor search.
\ifabs{}{Usually the data structure lower bounds are proved for a decision version of the problem. For nearest neighbor search, a natural decision version is the $\lambda$-\emph{near} neighbor problem $\lambda$-$\mathsf{NN}$. However, it is folklore that with both approximation and randomization allowed, $\lambda$-$\mathsf{NN}$ can be solved within $O(1)$-probe on a table of polynomial size. So to prove a nontrivial lower bound in this case, one must stay with the search problem. In~\cite{CharkReg}, this is done by a reduction from the longest prefix matching $\mathsf{LPM}$, a problem that critically captures the nature of \emph{searching} for the \emph{nearest} neighbors. }In~\cite{CharkReg}, a lower bound is proved for $\mathsf{LPM}$ by interpreting a data structure as a communication protocol and applying round eliminations to the communication protocol, a technique that can be traced back to~\cite{ajtai1988lower,miltersen1995data}.
Our main observation is that $k$ rounds of cell-probes can be represented as $2k$ rounds of communications. Although the observation is straightforward, to prove our lower bound we have to apply the techniques of~\cite{CharkReg} to adapt to non-uniform message sizes in different rounds, a setting which was rarely considered in the context of asymmetric communication complexity for data structure lower bounds. More critically, in order to get the $1/k$ exponent in our $\Omega\left(\frac{1}{k}(\log d)^{1/k}\right)$ lower bound, we are forced to exploit the round elimination of~\cite{CharkReg}. 
In fact, assuming $k=O\left(\frac{\log\log d}{\log\log\log d}\right)$, a lower bound with form $\Omega\left(k+\frac{1}{k^b}(\log d)^{a/k}\right)$ for any constants $a,b>0$ is enough to imply the optimal $\Omega\left(\frac{\log\log d}{\log\log\log d}\right)$ lower bound in Theorem~\ref{theorem-CharkReg}, whereas for our result, these constants $a,b$ matter a lot and require much delicacy in the round elimination argument.


\begin{suppress}
\paragraph{Related works.}

Our work is a theoretical study of the power of adaptiveness in approximate nearest neighbor search problem. On polynomial space, it gives almost tight results. Adaptiveness substantially reduces the complexity even for very weak algorithm (see our first upper bound). Besides, the complexity undergoes a ``threshold phenomenon'' when $k=\Theta(\log\log d/\log\log\log d)$.\todo{need more explanation?} Moreover, the \emph{nearest} and \emph{search} play critical roles in the complexity of approximate nearest neighbor search since we can see that the problem will have a $O(1)$ solution when it is weakened to approximate near neighbor search problem.
\par
The techniques in this work is derived from Chakrabarti and Regev's work~\cite{CharkReg}. In~\cite{CharkReg}, Chakrabarti \& Regev combined a series of lemma in communication complexity to deduce a round elimination lemma, then induce a contradiction with the round elimination lemma, which will prove desired lower bound. In this work we follow the same routine with a generalization to obtain our result. 
\par
The rest part may refer my story in the last section.

\begin{itemize}
	\item LSH from~\cite{IM98}. $s=O(dn+n^{1+\rho})$, $t=O(dn^\rho)$. $\rho\leq1/\gamma$.
	\item LSH from~\cite{DIIM04}. $\ell_2$-norm. $\rho<1/\gamma$.
	\item LSH from~\cite{panigrahy2006entropy}. $\ell_2$-norm. $s=\tilde{O}(n)$, $t=\tilde{O}(dn^\rho)$. $\frac{\rho}{\gamma}\approx 2.09$.
	\item LSH from~\cite{andoni2006near}. $\ell_2$-norm. $s=O(dn+n^{1+\rho})$, $t=O(dn^\rho)$. $\rho=1/\gamma^2+o(1)$.
	\item Lower bound of LSH from~\cite{MNP07}. For $\ell_p$-norm, $\rho\geq\frac{0.462}{\gamma^p}$.
\item
Optimal lower bounds for locality sensitive hashing(except when q is tiny) Ryan O'Donnell, Yi Wu, Yuan Zhou:
lower bound: $(r,cr,p,q)$-sensitive hash family, $\rho=ln(1/p)ln(1/q)$, $\rho$ must be at least $1/c$(minus $o_d(1)$) for hamming space.
\item
Indyk and Motwani showed that for each $c\geq1$, the extremely simple family achieves $\rho\leq1/c$.
\item
Beyond Locality-Sensitive Hashing: c-approximate near neighbor problem in the Euclidean space, $O_c(dn^{\rho})$ query time and $O_c(n^{1+\rho}+nd)$ space, where $\rho\leq7/(8c^2)+O(1/c^3)+o_c(1).$ a data structure for the Hamming space and $l_1$ norm with $\rho\leq7/(8c)+O(1/c^{3/2})+o_c(1)$.
\item
Optimal Data-Dependent Hashing for Approximate Near Neighbors: For an $n$-point dataset in a $d$-dimensional space our data structure achieves query time $O(d\cdot n^{\rho+o(1)})$ and space $O(n^{1+\rho+o(1)}+d\cdot n)$, where $\rho=\frac{1}{2c^2-1}$ for the Euclidean space and approximation $c>1$. For the Hamming space, $\rho=\frac{1}{2c-1}$.
\item
Approximate Nearest Neighbor: Towards Removing the Cure of Dimensionality: results for the approximate nearest problem under the $l_s$ norms for $s\in[1,2]$:

A deterministic data structure for $(1+\gamma)(1+\varepsilon)$-NN, for any constant $\gamma>0$, with space and preprocessing time bounded by $O(n\log^2 n)\times O(1/\varepsilon)^d$, and query time $O(d\log n)$.

A randomized data structure for $(1+\gamma)(1+\varepsilon)$-NN, for any constant $\gamma>0$, with space and preprocessing time bounded by $n^{O(\log(1/\varepsilon)\varepsilon^2)}$, and query time polynomial in $d$, $\log n$ and $1/\varepsilon$.

A randomized data structure for $c(1+\gamma)$-NN, for any constant $\gamma>0$, with space that is sub-quadratic in $n$ and query time sub-linear in $n$. Specifically, the data structure uses space $O(dn+n^{1+1/c}\log^2 n\log\log n)$, and has query time of $O(dn^{1/c}\log^2 n \log\log n)$
\end{itemize}

我们是否应该在最后的定理中省略掉对$\gamma$的界(i.e. $3\leq\gamma\ll d$)？因为前文和原文中都完全略去了这一点。

\end{suppress}

\section{Preliminaries}\label{section-prelim}
\ifabs{\paragraph{Approximate nearest neighbor search}}{\paragraph{Approximate nearest neighbor search:}}
We consider the problem of \concept{approximate nearest neighbor search} in the $d$-dimensional Hamming space $\ANNS$.
Let $\gamma>1$ be fixed. 
We are given a \concept{database} $B$ which contains $n$ points from the $d$-dimensional Hamming cube $\{0,1\}^d$. The database is preprocessed into a data structure (called the \concept{table}). Then given any \concept{query} point $x\in\{0,1\}^d$, the algorithm must access the data structure and output a database point $y\in B$ which is a \concept{$\gamma$-approximate nearest neighbor} of $x$ in $B$, where a point $y\in B$ is called a $\gamma$-approximate nearest neighbor of $x$ in $B$ if $\mathrm{dist}(x,y)\le \gamma\cdot\min_{z\in B}\mathrm{dist}(x,z)$, where $\mathrm{dist}(x,y)$ denotes the Hamming distance between $x$ and $y$.

Abstractly, a \concept{data structure problem} can be represented as a relation $\rho\subseteq\mathscr{A}\times\mathscr{B}\times\mathscr{C}$, where $\mathscr{A}, \mathscr{B}$, and $\mathscr{C}$ specify the universes for queries, databases, and answers, respectively. Given a query $x\in \mathscr{A}$ to a database $B\in\mathscr{B}$, an answer $z\in\mathscr{C}$ is correct if $(x,B,z)\in\rho$. In particular, for approximate nearest neighbor search, $\mathscr{A}=\mathscr{C}=\{0,1\}^d$, $\mathscr{B}={\{0,1\}^d\choose n}$, and 
\ifabs{
\begin{align*}
\ANNS=\big{\{}
&(x,B,z)\in\mathscr{A}\times\mathscr{B}\times\mathscr{C}\mid\\
& z\in B\land \forall y\in B, \mathrm{dist}(x,z)\leq\gamma\cdot\mathrm{dist}(x,y)\big{\}}.
\end{align*}
}{
\begin{align*}
\ANNS=\left\{(x,B,z)\in\mathscr{A}\times\mathscr{B}\times\mathscr{C}\mid z\in B\land \forall y\in B, \mathrm{dist}(x,z)\leq\gamma\cdot\mathrm{dist}(x,y)\right\}.
\end{align*}
}

\ifabs{\paragraph{The cell-probe model}}{\paragraph{The cell-probe model.}} We adopt Yao's cell-probe model~\cite{yao1981should} for static data structures. A \concept{cell-probing scheme} $(\mathcal{A},\mathcal{T})$ for a data structure problem $\rho\subseteq\mathscr{A}\times\mathscr{B}\times\mathscr{C}$
consists of a \concept{cell-probing algorithm} $\mathcal{A}$ and a code (sometimes called the \concept{table structure}) $\mathcal{T}$. Each database $B\in\mathscr{B}$ is mapped by the code $\mathcal{T}:\mathscr{B}\to\Sigma^s$ to a codeword (called a \concept{table}) $T_B\in\Sigma^{[s]}$ of $s$ symbols from the alphabet $\Sigma=\{0,1\}^w$. The interpretation is that each database $B$ is preprocessed and stored as a {table} $T_B$ consisting of $s$ table \concept{cells} (also called a \concept{word}), with each cell storing $w$ bits.  Given an address $i\in[s]$ we use $T_B[i]$ to denote the content of the $i$-th cell in table $T_B$. 
Given a query $x\in\mathscr{A}$, the cell-probing algorithm $\mathcal{A}$ must output a correct answer $z\in\mathscr{C}$ such that $(x,B,z)\in\rho$, after accessing the table $T_B$ adaptively for $t$ times, each time reading the content of one table cell. This action of reading the content of one table cell by the cell-probing algorithm is usually called as making a \concept{cell-probe}.

The complexity of a cell-probing scheme is captured by three parameters: namely, the \concept{table size} $s$, the \concept{word size} $w$, and the time cost or \concept{cell-probe complexity} $t$. 

\ifabs{\paragraph{Cell-probe model with limited adaptivity}}{\paragraph{Cell-probe model with limited adaptivity:}}
In this work, we refine the cell-probe model by considering the rounds of parallelizable cell-probes in cell-probing algorithms. Formally, a \concept{$k$-round} cell-probing algorithm $\mathcal{A}$ can be described by $k$ \concept{lookup functions} $L_1,L_2,\ldots,L_k$ and one \concept{truth table} $A$. Each lookup function $L_i$ maps the query $x$ and the contents of the table cells probed before round $i$, to a sequence of addresses indicating the set of table cells to probe in round $i$. In the beginning, $L_1(x)=(p^1_1,p^1_2,\ldots,p^1_{t_1})\in[s]^{t_1}$ for some $t_1>0$, and for general $1\le i\le k$:
\ifabs{
\[
L_i\left(x, \left\langle p^j_\ell, T_B[p^j_\ell]\right\rangle_{ 1\le j<i\atop 1\le \ell\le t_j}\right)=(p^i_1,p^i_2,\ldots,p^i_{t_i})\in[s]^{t_i},
\]
for some $t_i>0$, 
}{
\[
L_i\left(x, \left\langle p^j_\ell, T_B[p^j_\ell]\right\rangle_{ 1\le j<i\atop 1\le \ell\le t_j}\right)=(p^i_1,p^i_2,\ldots,p^i_{t_i})\in[s]^{t_i}, \quad\text{ for some }t_i>0,
\]
}
so that at round $i$, the algorithm makes $t_i$ \emph{parallel} cell-probes to the the cells $p^i_\ell$, $1\le \ell \le t_i$. 
And finally, the truth table $A$ maps the contents of all the probed cells $\langle p^j_\ell, T_B[p^j_\ell]\rangle_{ 1\le j\le k\atop 1\le \ell\le t_j}$, to a correct answer $z$ satisfying that $(x, B,z)\in\rho$. 
The cell-probe complexity is given by $t=t_1+t_2+\cdots+t_k$. 
This formulation includes the standard definition of cell-probing scheme as a special case when $t_1=\cdots =t_k=1$. 

\ifabs{\paragraph{Public-coin vs.~private-coin cell-probing schemes}}{\paragraph{Public-coin vs.~private-coin cell-probing schemes:}}
In a (private-coin) randomized cell-probing scheme, the table is prepared by a code $\mathcal{T}$ deterministically as before, but the cell-probing algorithm $\mathcal{A}$ is a randomized algorithm. This can be considered as that the deterministic lookup functions $L_1,L_2,\ldots,L_k$ as well as the truth table $A$ also take a sequence of random bits $r\in\{0,1\}^*$ as part of the input.  We say we have a randomized cell-probing scheme $(\mathcal{A},\mathcal{T})$ for a data structure problem $\rho\subseteq\mathscr{A}\times\mathscr{B}\times\mathscr{C}$ if for every query $x\in\mathscr{A}$ and every database $B\in\mathscr{B}$, the cell-probing algorithm outputs a correct answer $z\in\mathscr{C}$ such that $(x,B,z)\in\rho$ with probability at least $2/3$.
The constant $2/3$ is rather arbitrary. Note that for problems such as approximate nearest neighbor search, where once the query $x$ is known, a monotone order of the correctness between multiple answers is fixed, any positive constant success probability is enough: we can boost it to any constant accuracy $1-\epsilon$ by independent repetition of the cell-probing algorithm for constant many times \emph{in parallel}, which will keep the asymptotic cell-probe complexity and the number of rounds of parallel cell-probes.

In this paper, all of our upper bounds will be presented first as \concept{public-coin} randomized cell-probing schemes. For a public-coin randomized cell-probing scheme, the sequence of random bits $r\in\{0,1\}^*$ is shared between the cell-probing algorithm $\mathcal{A}$ and the table structure $\mathcal{T}$, where the table $T_B^r$ is now determined by both the database $B$ and the random bits $r$. 
This makes no change to the family of data structures of polynomial size: by Newman's theorem~\cite{newman1991private}, a public-coin cell-probing scheme can be transformed to a standard randomized cell-probing scheme, where the randomness is private to the cell-probing algorithm. 

\begin{lemma}\label{lemma-public-vs-private}
If there is a $k$-round public-coin randomized cell-probing scheme for a data structure problem $\rho\subseteq\mathscr{A}\times\mathscr{B}\times\mathscr{C}$ with table size $s$, word size $w$, and cell-probe complexity $t$, then there is a $k$-round randomized cell-probing scheme for $\rho$ with table size $(\log|\mathscr{A}|+\log|\mathscr{B}|+O(1))s$, word size $w$, and cell-probe complexity $t$.
\end{lemma}
\begin{proof}
The proof is similar to the proof of Lemma 6.5 in~\cite{CharkReg}, with the observation that the translation there also preserves the number of rounds. Without loss of generality, we assume that for every query to every database, the $k$-round public-coin randomized cell-probing scheme returns a correct answer except with an error probability at most $1/4$. The $k$-round public-coin randomized cell-probing scheme can be seen as a $k$-round public coin randomized communication protocol between Alice for the cell-probing algorithm and Bob for the table, where Bob is non-adaptive thus is only response to each individual message received in the current round according to its input $B$ in a consistent way (as a code). By Newman's theorem, the number of public random bits can be reduced to $\ell=\log(\log|\mathscr{A}|+\log|\mathscr{B}|+O(1))$ while the error probability is raised to $1/3$. This does not change the structure of the protocol, so it can be translated back to a $k$-round public-coin randomized cell-probing scheme for $\rho$ with the same time and space complexity as before and with $\ell$ public random bits. We create a table $T_B^r$ for every possible sequence of random bits $r\in\{0,1\}^\ell$ according to the public-coin cell-probing scheme. This gives us a big table of size $s\cdot2^\ell=s(\log|\mathscr{A}|+\log|\mathscr{B}|+O(1))$, and the random bits is made private to the cell-probing algorithm. 
\end{proof}

\ifabs{\noindent\textbf{Notations.}}{\paragraph{Notations.}} We use $\mathrm{dist}(\cdot,\cdot)$ to denote Hamming distance. We write $\log$ for binary logarithm and $\ln$ for natural logarithm.

\begin{suppress}
\noindent\textbf{Approximate Nearest Neighbor Search:} Let $\gamma>1$ be some fixed constants, we define Approximate Nearest Neighbor problem $\ANNS$ as a data structure query problem which is given by 
\[
	\mathscr{A}=\{0,1\}^d,\mathscr{B}=\binom{\{0,1\}^d}{n},\mathscr{C}=\{0,1\}^d
\]
\begin{align*}
	\rho=\{(x,y,z)\in\mathscr{A}\times\mathscr{B}\times\mathscr{C}:z\in y\land(\forall z'\in y(\mathrm{dist}(x,z)\leq\gamma\cdot\mathrm{dist}(x,z')))\},
\end{align*}
where ``$\mathrm{dist}$'' denotes Hamming distance in $\{0,1\}^d$.

For a universe of queries $\mathscr{A}$, a universe of databases $\mathscr{B}$, a universe of answers $\mathscr{C}$, and a relation $\rho\subseteq\mathscr{A}\times\mathscr{B}\times\mathscr{C}$, a \textbf{data structure query problem}, as we know, is a problem which gives us a database $y\in\mathscr{B}$ and requires us to build a data structure $D$ such that given any $x\in\mathscr{A}$ we can interact with the data structure $D$ and output a $z\in\mathscr{C}$ such that $(x,y,z)\in\rho$. In this work we look at \textbf{Approximate Nearest Neighbor Search problem} and \textbf{Approximate Near Neighbor Search problem} over the $d$-dimensional hypercube in \textbf{Randomized Cell Probe Model with Public-coin}.
\par
\noindent\textbf{Cell Probe Model and Cell Probe Complexity:} Cell Probe Model is a model of computation which is described for first time in Andrew Yao's work~\cite{yao1981should}. In cell probe model, a data structure $D$ be consider as a table consisting of $s$ cells each of which holds $w$ bits. The complexity of an algorithm for a data structure query problem is measured by a tuple $(s,w,t)$: $s$, the size of the table; $w$, the size of each cell; and $t$, the number of cells be accessed in worst case. Note that the cell probe model consider all operations are free except the cell access.
\par
\noindent\textbf{Cell Probe Scheme:} An algorithm for data structure query problem in cell probe model is called a cell probe scheme. A cell probe scheme is consited of two phases: the preprocessing phase and the query phase. In preprocessing phase, the scheme constructs a table according to the database $y$. 
In query phase, the scheme forgot the database $y$ as well as the content of table, then probes at most $t$ cells according to the query $x$, and output the answer $z$.
In this work, we focus on another complexity measure, the number of rounds we probe the cells. (Note that we can access several cells in one round.) Thus we measure the cell probe complexity in this work by a tuple $(s,w,k,t)$, which means the table contains $s$ cells, each cell contains at most $w$ bits, we probe the cells at most $k$ rounds, and probe at most $t$ cells in total.
\par
\noindent\textbf{Randomized Cell Probe Scheme with Public-coin:} Be similar with other computational problems, the cell probe schemes can be randomized. A randomized cell probe scheme allow randomization in preprocessing phase as well as query phase, i.e. we can randomizedly construct the table and randomizedly probe the cells. Also be similar with other computational problems, the randomness in randomized cell probe scheme can be private-coin or public-coin. For a randomized cell probe scheme with public-coin, we consider the random sources in the preprocessing phase and the query phase are same. Formally, let $RP(i)$ be the $i$-th output of the random source in preprocessing phase and $RQ(i)$ be the $i$-th output of the random source in query phase, then $\forall i,RP(i)=RQ(i)$.
We measure the randomized cell probe complexity in this work by a tuple $(s,w,k,t,\epsilon)$, which means the table contains $s$ cells, each cell contains at most $w$ bits, we probe the cells at most $k$ rounds, probe at most $t$ cells in total, and for any input the scheme outputs correct answer with probability at least $1-\epsilon$. When $\epsilon$ is not specified we assume that it is $\frac{1}{8}$.
\par
\noindent\textbf{Approximate Nearest Neighbor Search:} Let $\gamma>1$ be some fixed constants, we define Approximate Nearest Neighbor problem $\ANNS$ as a data structure query problem which is given by 
\[
	\mathscr{A}=\{0,1\}^d,\mathscr{B}=\binom{\{0,1\}^d}{n},\mathscr{C}=\{0,1\}^d
\]
\begin{align*}
	\rho=\{(x,y,z)\in\mathscr{A}\times\mathscr{B}\times\mathscr{C}:z\in y\land(\forall z'\in y(\mathrm{dist}(x,z)\leq\gamma\cdot\mathrm{dist}(x,z')))\},
\end{align*}
where ``$\mathrm{dist}$'' denotes Hamming distance in $\{0,1\}^d$.
\par
\noindent\textbf{Approximate Near Neighbor Search Problem:} Let $\lambda>0,\gamma>1$ be some fixed constants, we define $\lambda$-Aprroximate Near Neighbor problem $\lANNS$ as a data structure query problem which is given by
\[
	\mathscr{A}=\{0,1\}^d,\mathscr{B}=\binom{\{0,1\}^d}{n},\mathscr{C}=\{0,1\}^d\cup\mathrm{NO}
\]
\begin{align*}
	&\rho=\{(x,y,z)\in\mathscr{A}\times\mathscr{B}\times\mathscr{C}:\\
	&(\nexists z'\in y(\mathrm{dist}(x,z')\leq\lambda)\lor z\in y(\mathrm{dist}(x,z)\leq\gamma\lambda))\land(\exists z'\in y(\mathrm{dist}(x,z')\leq\gamma\lambda)\lor z=\mathrm{NO})\},
\end{align*}
where ``dist" denotes Hamming distance in $\{0,1\}^d$.
\end{suppress}

\section{Approximate nearest neighbor search in $k$ rounds}\label{section-upper-bound}
In this section, we will give two algorithms that solve the approximate nearest neighbor search problem $\ANNS$ within $k$ rounds on a table of size $n^{O(1)}$ and word size $O(d)$:
\begin{enumerate}
\item
a simple $k$-round randomized cell-probing scheme with $O(k(\log d)^{1/k})$ cell-probes;
\item 
a more sophisticated $k$-round randomized cell-probing scheme with $O(k+(\frac{1}{k}\log d)^{c/k})$ cell-probes, for any constant $c>2$.
\end{enumerate}
\ifabs{
In Appendix~\ref{appendix-ann}, we will also describe a folklore result in the current framework, which shows that if the problem is relaxed a little to the approximate \emph{near}-neighbor search problem (instead of the \emph{nearest} neighbor search), then on a table of polynomial size, the problem can be solved with $O(1)$ cell-probes by a \emph{non-adaptive} (i.e.~1-round) randomized cell-probing scheme.
}{
We will also include a folklore result in the current framework to show that if the problem is relaxed a little to the approximate \emph{near}-neighbor search problem (instead of the \emph{nearest} neighbor search), then on a table of polynomial size with word size $O(d)$, the problem can be solved with $O(1)$ cell-probes by a \emph{non-adaptive} (i.e.~1-round) randomized cell-probing scheme.
}

\ifabs{\paragraph{Public-coin vs.~private-coin in the cell-probe model}}{\paragraph{Public-coin vs.~private-coin in the cell-probe model.}}
All our three algorithms will be first presented as \emph{public-coin} cell-probing schemes, where the random bits are shared between the cell-probing algorithm and the table, and then transformed by Lemma~\ref{lemma-public-vs-private} to the standard randomized cell-probing schemes, where the random bits are private to the cell-probing algorithm, with the same round and cell-probe complexity and a polynomial overhead to the table size.   
In particular, for $\ANNS$ we have the following proposition.

\begin{proposition}\label{propo-public-vs-private-ANNS}
If there is a $k$-round public-coin randomized cell-probing scheme for $\ANNS$ with table size $s$, word size $w$, and cell-probe complexity $t$, there exists a $k$-round randomized cell-probing scheme for $\ANNS$ with table size $O(dn\cdot s)$, word size $w$, and cell-probe complexity $t$.
\end{proposition}

\ifabs{\paragraph{Common setup for the algorithms}}{\paragraph{Common setup for the algorithms.}}

\begin{suppress}
\begin{definition}
	Let $r$ be a positive integer and $\lambda$ a real number with $\lambda\geq1$. We define $\mathcal{V}_\lambda$ to be the distribution of a random $d$-coordinate row verctor in which each coordinate is independently chosen to be $1$ with probability $1/(4\lambda)$ and $0$ otherwise. We define $\mathcal{M}^r_\lambda$ to be the distribution of a random $r\times d$ matrix where each row is independently chosen from distribution $\mathcal{V}_\lambda$.
\end{definition}

\begin{lemma}[refined from Lemma 6.2 in~\cite{CharkReg}]
\label{lem:vector equality}
	Let $\lambda\geq1$ and $\gamma>1$. Then, there exist two numbers $\delta_1(\lambda)<\delta_2(\lambda,\gamma)$, both in $[0,1]$, such that $\delta_2(\lambda,\gamma)-\delta_1(\lambda)=(\Delta-\Delta^\gamma)/2$, where $\Delta$ depends only on $\lambda$ and satisfies $1/2\leq\Delta\leq1/\sqrt{e}$, and such that, for all $d\geq1$ and for all $x,z\in\{0,1\}^d$,
	\begin{align*}
		\mathrm{dist}(x,z)\leq\lambda &\implies\Pr[Yx\neq Yz]\leq\delta_1(\lambda),\\
		\mathrm{dist}(x,z)>\gamma\lambda &\implies\Pr[Yx\neq Yz]>\delta_2(\lambda,\gamma),
	\end{align*}
	where $Y$ is a random row verctor drawn from distribution $\mathcal{V}_\lambda$.
\end{lemma}
\begin{proof}
The majority of the proof is going to be the same as that of Lemma 6.2 in~\cite{CharkReg}, except for a more carful quantitive analysis of the gap between $\delta_2$ and $\delta_1$.

	Consider the following equivalent way of choosing $Y$: first choose a set $\mathbf{C}\subseteq[d]$ where each integer in $[d]$ is put in $\mathbf{C}$ independently with probability $1/(2\lambda)$. Then, for each $i\in\mathbf{C}$ let the $i$-th coordinate of $Y$ be chosen uniformly from $\{0,1\}$. For $i\notin\mathbf{C}$, set the $i$-th coordinate of $Y$ to zero. Let $z\in\{0,1\}^d$ be arbitrary and let $h=\mathrm{dist}(x,z)$. If $\mathbf{C}$ does not contain any of the coordinates on which $x$ and $z$ differ, then clearly $Yx=Yz$. This happens with probability $(1-1/(2\lambda))^h$. Otherwise, if $\mathbf{C}$ contains at least one of the coordinates on which $x$ and $y$ differ, the probability that $Yx\neq Yz$ is precisely $1/2$ (note that we perform matrix multiplication on GF(2)). Hence,
	\[
		\Pr[Yx\neq Yz]=\frac{1}{2}\left(1-\left(1-\frac{1}{2\lambda}\right)^h\right).
	\]
	It can be seen that this is a monotonically increasing function of $h$ and that by plugging in $\lambda$ and $\gamma\lambda$ for $h$ one obtains $\delta_2$ and $\delta_1$ as follows.

	Denote $\sigma=\delta_2-\delta_1$. We have
	\begin{align*}
		2\sigma	&=\left(1-\frac{1}{2\lambda}\right)^\lambda-\left(1-\frac{1}{2\lambda}\right)^{\gamma\lambda}\\
			   &=\left(\left(1-\frac{1}{2\lambda}\right)^{2\lambda}\right)^{1/2}-\left(\left(1-\frac{1}{2\lambda}\right)^{2\lambda}\right)^{\gamma/2}.
	\end{align*}
	Let $\Delta=(1-\frac{1}{2\lambda})^\lambda$. Note that $\Delta^2$ is monotonically increasing when $2\lambda\geq1$, and $\lim_{\lambda\to+\infty}\Delta^2=1/e$. Recall that $\lambda\geq1$. Consequently $1/2\leq\Delta<1/\sqrt{e}$. Hence
	\begin{align*}
		\sigma=(\Delta-\Delta^\gamma)/2,
	\end{align*}
	where $\Delta$ is a function of $\lambda$ satisfying $1/2\leq\Delta\leq1/\sqrt{e}$.
\end{proof}

The following lemma is a  consequence of Lemma~\ref{lem:vector equality} by the Chernoff bound and was proved in~\cite{CharkReg}. 
\begin{lemma}[Lemma 6.3 in~\cite{CharkReg}]\label{lem:threshold}
	Let $\lambda\geq1$ and $\gamma>1$. Define $\delta(\lambda,\gamma)=(\delta_1(\lambda)+\delta_2(\lambda,\gamma))/2$, where $\delta_1,\delta_2$ are as in Lemma \ref{lem:vector equality}. Then, for all $d\geq1$, all $u,v\in\{0,1\}^d$, and all $r\geq1$,
	\begin{align*}
		\mathrm{dist}(u,v)\leq\lambda &\implies\Pr[\mathrm{dist}(Mu,Mv)>\delta(\lambda,\gamma)\cdot k]\leq\exp(-(\delta_2-\delta_1)^2r/2),\\
		\mathrm{dist}(u,v)>\gamma\lambda &\implies\Pr[\mathrm{dist}(Mu,Mv)\leq\delta(\lambda,\gamma)\cdot k]\leq\exp(-(\delta_2-\delta_1)^2r/2),
	\end{align*}
	where $M$ is a random matrix drawn from distribution $\mathcal{M}_\lambda^r$.
\end{lemma}
\end{suppress}

We consider only constant approximation ratio $\gamma>1$, so without lost of generality, we can assume that $\gamma<4$, since for larger $\gamma$ our algorithms will only have better approximation. Let $\alpha\triangleq\sqrt{\gamma}$, and hence $1<\alpha<2$. 
Let $x\in\{0,1\}^d$ denote the query point and $B\subseteq\{0,1\}^d$, $|B|=n$, denote the database. We always assume that $n>d$. 
For $0\le i\le\lceil\log_\alpha d\rceil\}$, let 
\begin{align}
B_i=\{y\in B\mid \mathrm{dist}(x,y)\le \alpha^i\},\label{eq:B-i}
\end{align} 
be the set of all database points within distance $\alpha^i$ of $x$.

\begin{definition}\label{definition-upper-bound-common}
Let $c_1,c_2>64/(1-\mathrm{e}^{(1-\alpha)/2})^2$ be constants and $1<s<\ln\ln n$. 
For $0\le i\le \lceil\log_\alpha d\rceil$, let $M_i, N_i$ be the independent random Boolean matrices such that each entry is sampled \emph{i.i.d.}~from $\mathrm{Bernoulli}(\frac{1}{4\alpha^i})$, with $M_i$ of size $(c_1\log n)\times d$ and $N_i$ of size $(\frac{c_2}{s}\log n)\times d$.
For $0\le j\le i\le \lceil\log_\alpha d\rceil$, we define the sets
\ifabs{
\begin{align}
C_i 
&=
\left\{z\in B\mid \mathrm{dist}(M_ix,M_iz)\leq\delta(\alpha^i,\alpha) c_1\log n\right\},\label{eq:C-i}\\
D_{i,j} 
&=
\left\{z\in C_i\mid \mathrm{dist}(N_jx,N_jz)\leq\delta(\alpha^j,\alpha) \frac{c_2}{s}\log n\right\}, \label{eq:D-i}
\end{align}
}{
\begin{align}
C_i 
&=
\left\{z\in B\mid \mathrm{dist}(M_ix,M_iz)\leq\delta(\alpha^i,\alpha)\cdot c_1\log n\right\},\label{eq:C-i}\\
D_{i,j} 
&=
\left\{z\in C_i\mid \mathrm{dist}(N_jx,N_jz)\leq\delta(\alpha^j,\alpha)\cdot (c_2\log n)/s\right\}, \label{eq:D-i}
\end{align}
}
where $\delta(\beta,\alpha)=\frac{1}{2}\left(1-\frac{1}{2\beta}\right)^{\beta}\left[1-\left(1-\frac{1}{2\beta}\right)^{(\alpha-1)\cdot\beta}\right]$.
\end{definition}

The following lemma proved in~\cite{CharkReg} shows that $C_i$ are approximations of the balls $B_i$, and $D_i$ are also approximations in a weaker sense.
\begin{lemma}[Chakrabarti and Regev~\cite{CharkReg}]\label{lemma-upper-bound-assumption}
The following events hold simultaneously with probability at least $3/4$:
\begin{enumerate}
\item
$B_i\subseteq C_i\subseteq B_{i+1}$ for all $i$.
\item
For all $0\le j\le i\le \lceil\log_\alpha d\rceil$, 
at most a fraction $n^{-1/s}$ of $B_j$ is not in $D_{i,j}$ and at most a fraction $n^{-1/s}$ of $C_i\setminus B_{j+1}$ is in $D_{i,j}$.
\end{enumerate}
\end{lemma}

\subsection{A simple $k$-round protocol for $\mathsf{ANNS}$}\label{subsection-simple-upper-bound}
\begin{theorem}[Theorem~\ref{theorem-simple-upper-bound-informal}, restated]\label{theorem-simple-upper-bound}
Let $\gamma>1$ be any constant. For $n>d$ and $k\ge 1$, $\ANNS$ has a $k$-round randomized cell-probing scheme with $O\left(k(\log d)^{1/k}\right)$ cell-probes, 
table size $n^{O(1)}$ and word size $O(d)$.
\end{theorem}

As mentioned earlier, the solution will be presented as a \emph{public-coin} cell-probing scheme, which by Proposition~\ref{propo-public-vs-private-ANNS}, is then transformed to a standard randomized cell-probing scheme with the same cell-probe complexity and number of rounds on a table of polynomial size.

Without lost of generality, assume that $\gamma<4$ and let $\alpha\triangleq\sqrt{\gamma}$. 
Let $x\in\{0,1\}^d$ denote the query point and $B\subseteq\{0,1\}^d$, $|B|=n$, denote the database. Recall that $B_i$, as defined in~\ref{eq:B-i}, are the sets of all database points within distance $\alpha^i$ of $x$.

There are two degenerate cases. The first case is when $B_0$ is not empty, which means $x\in B$. This case can be solved as a membership query of $x$ in the set $B$, by the perfect hashing with 1 cell-probe to a table of size $O(n^2)$, with the random hash function as public randomness. The second degenerate case is when $B_1$ is not empty, which means the query point $x$ is within distance 1 from $B$. This can also be solved as a membership query of $x$ in the 1-neighborhood $N_1(B)=\{y\in\{0,1\}^d\mid \exists z\in B, \mathrm{dist}(y,z)\le 1\}$ of $B$, which contains at most $(d+1)n$ points, by the same method, using 1 cell-probe to a table of quadratic size with public randomness. 

Note that these two instances of perfect hashing can run separately and in parallel to each other, and to the main data structure solving the non-degenerate cases, so that if a query $x$ finds itself within $B$ or within distance 1 from $B$, then the algorithm terminates and outputs the nearest neighbor. This will cost a polynomial addition to the table size and 2 more queries in the first round, but make no change to the number of rounds. For the rest, we can assume the following.
\begin{assumption}\label{assume-0}
$B_0=B_1=\emptyset$.
\end{assumption}

The goal of the main data structure is to find an $i$ such that $B_i$ is empty but $B_{i+2}$ is not and output a point in $B_{i+2}$, assuming that $B_0=B_1=\emptyset$. Such a point is clearly a $\gamma$-approximate nearest neighbor of $x$. 

For $0\le i\le \lceil\log_\alpha d\rceil$, let $M_i$ be the random $(c_1\log n)\times d$ matrices sampled independently as in Definition~\ref{definition-upper-bound-common} and $C_i\subseteq B$ the subsets of database points constructed from $M_i$ as in Definition~\ref{definition-upper-bound-common}. The random matrices $M_i$ are treated as the public randomness shared between the cell-probing algorithm and the table.
The table who possesses the database $B$ may construct $C_i$ from $M_ix$ for every possible $x\in\{0,1\}^d$ (in fact, for every $M_ix\in\{0,1\}^{c_1\log n}=[n^{c_1}]$), while the cell-probing algorithm who possesses the query point $x$ may compute the product $M_ix$ from the actual query point $x$.

By Lemma~\ref{lemma-upper-bound-assumption}, the following assumption holds with probability at least $3/4$:
\begin{assumption}\label{assume-1}
$B_i\subseteq C_i\subseteq B_{i+1}$ for all $i$.
\end{assumption}
With this assumption, the algorithm only needs to find an $i$ such that $C_{i}\neq\emptyset$ but $C_{i-1}=\emptyset$. Since $B_{i-1}\subseteq C_{i-1}=\emptyset$ and $B_{i+1}\supseteq C_{i}\neq\emptyset$, any  point in $C_{i}$ is a $\gamma$-approximation nearest neighbor of $x$.

\ifabs{\paragraph{Table construction}}{\paragraph{Table construction.}}
We construct $\lceil\log_\alpha d\rceil +1$ tables $T_0,\ldots, T_{\lceil\log_\alpha d\rceil}$. Each table $T_i$ contains $2^{c_1\log n}=\mathrm{poly}(n)$ many cells, where each cell corresponds to a string $j\in\{0,1\}^{c_1\log n}$, so the total number of cells in all these tables is a polynomial of $n$. Here $c_1$ is the constant factor in the number of rows of $M_i$.
Due to the public randomness, the table contents may depend on both the database $B$ and the public random matrices $M_i$.

For every $0\le i\le \lceil\log_\alpha d\rceil$ and every $j\in\{0,1\}^{c_1\log n}$, the content of the $j$-th cell $T_i[j]$ in the $i$-th table $T_i$ is given as follows: 
\begin{itemize}
\item
If there exists a datapoint $z\in B$ such that $\mathrm{dist}(j,M_iz)\leq \delta(\alpha^i,\alpha)\cdot c_1\log n$, the cell $T_i[j]$ stores an arbitrary one of such $z$.
\item
If otherwise there is no such datapoint, $T_i[j]$ stores a special symbol indicating the EMPTY.
\end{itemize}
Note that $M_ix\in\{0,1\}^{c_1\log n}$ is a valid address for the cells in a table $T_i$. 
And for every $i$, the table cell $T_i[M_ix]$ stores a point from $C_i$ if $C_i$ is not empty, or $T_i[M_ix]=\mathrm{EMPTY}$ if $C_i=\emptyset$.

\ifabs{\paragraph{Cell-probing algorithm}}{\paragraph{Cell-probing algorithm.}}
The algorithm possesses the query point $x$ and the public random matrices $M_i$.
Set $\tau=c'(\log d)^{1/k}$, for a constant $c'\ge\log_\alpha 4$ so that
\[
\tau\cdot\left(\frac{\tau}{2}\right)^{k-1} \geq \lceil\log_\alpha d\rceil.
\]
The cell-probing algorithm consists of at most $(k-1)$ \emph{shrinking rounds}, succeeded by one final \emph{completion round}. And if $k=1$, the algorithm is non-adaptive and just consists of a completion round. 
In every round the algorithm makes at most $\tau$ parallel cell-probes to the table. The total number of cell-probes is at most $(\tau-1)(k-1)+\tau=O(k(\log_\alpha d)^{1/k})$. The pseudocode of the cell-probing algorithm is given in Algorithm~\ref{alg:simple}.

\begin{algorithm}[!htp]
\ifabs{\caption{Simple $k$-round cell-probing algorithm}\label{alg:simple}
}{\caption{Simple $k$-round cell-probing algorithm for $\ANNS$}\label{alg:simple}}
\begin{algorithmic}
\State Set $\tau\gets (\log d)^{1/k}\log_\alpha4$;
\State initialize $l\gets 0$ and $u\gets\log_\alpha d$;
\While {$u-l\geq \tau$} \ifabs{}{\Comment{shrinking rounds}}
\State let $\rho(r)\gets\lfloor l+\frac{r}{\tau}(u-l)\rfloor$ for $0\le r\le\tau-1$;
\State retrieve $T_{\rho(r)}[M_{\rho(r)}x]$ for $1\leq r\leq \tau-1$;
\If {$\exists$ $1\leq r\leq \tau-1$ s.t. $T_{\rho(r)}[M_{\rho(r)}x]\neq \mathrm{EMPTY}$ }
\State let $r^*$ be the smallest such r;
\Else
\State $r^*\gets\tau$;
\EndIf
\State update $l\gets\rho(r^*-1)$ and $u\gets\rho(r^*)$;
\EndWhile
\State retrieve $T_i[M_ix]$ for all $l+1\leq i\leq u$; \ifabs{}{\Comment{completion round}}
\State $i^*\gets\min\{l+1\leq i\leq u:  T_i[M_ix]\neq \mathrm{EMPTY}\}$;
\State \Return $T_{i^*}[M_{i^*}x]$; 
\end{algorithmic}
\end{algorithm}

The algorithm finds a $\gamma$-approximate nearest neighbor of $x$ by a multi-way search: it maintains two integers $l$ and $u$, initially $l=0$ and $u=\lceil \log_\alpha d\rceil$. At each round $l$ and $u$ are updated, satisfying the invariant that $l<u$, $C_l=\emptyset$ and $C_u\neq\emptyset$. This invariant is satisfied initially since by Assumption~\ref{assume-0} and~\ref{assume-1} we have $C_0\subseteq B_1=\emptyset$ and $C_{\lceil \log_\alpha d\rceil}\supseteq B_{\lceil \log_\alpha d\rceil}=B$. For $0\le r\le \tau$, we denote $\rho(r)\triangleq\lfloor l+\frac{r}{\tau}(u-l)\rfloor$. The cell-probing algorithm proceeds as follows:
\begin{enumerate}
\item
In each shrinking round: the algorithm reads the contents of $T_{\rho(r)}[M_{\rho(r)}x]$ for all $1\leq r \leq \tau-1$ in parallel, and finds those $r$ such that $T_{\rho(r)}[M_{\rho(r)}x]\neq \mathrm{EMPTY}$, which means $C_{\rho(r)}\neq\emptyset$. Let $r^*$ be the smallest such $r$, or let $r^*=\tau$ if no such $r$ exists. Update $l$ to $\rho({r^*-1})$ and $u$ to  $\rho({r^*})$. The new gap between $l$ and $u$ is $\rho({r^*})-\rho({r^*-1})$, which is at most $(u-l)/\tau+1$.
\item
Once the gap $u-l$ drops below $\tau$, the algorithm enters the completion round: it reads the cells $T_{i}[M_{i}x]$ for all $l+1\le i\le u$ in parallel, finds the smallest $i$ such that $T_{i}[M_{i}x]\neq\mathrm{EMPTY}$, and outputs the point stored in that cell. Such $i$ must exist since we know $C_u\neq\emptyset$.
Note that the output point is from a nonempty $C_{i}$ such that $C_{i-1}=\emptyset$. With Assumption~\ref{assume-1}, it must be a $\gamma$-approximate nearest neighbor of $x$. 
\end{enumerate}
Note that in every shrinking round,  $l$ and $u$ are updated to $l'$ and $u'$ respectively so that $u'-l'\leq (u-l)/\tau+1\le 2(u-l)/\tau$ as long as $u-l\geq \tau$. And once $u-l< \tau$, the algorithm enters the completion round. Recall that $\tau\cdot (\tau/2)^{k-1}  \geq \lceil\log_\alpha d\rceil$. Hence, there can be at most $(k-1)$ shrinking rounds.

\subsection{A $k$-round protocol for $\mathsf{ANNS}$ for large $k$}\label{subsection-general-upper-bound}
\begin{theorem}[Theorem~\ref{theorem-general-upper-bound-informal}, restated]\label{theorem-general-upper-bound}
Let $\gamma>1$ and $c>2$ be any constants. 
For $n> d$ and $k> 5c^2/(c-2)$, $\ANNS$ has a $k$-round randomized cell-probing scheme with $O\left(k+\left(\frac{1}{k}\log d\right)^{c/k}\right)$ cell-probes,
  table size $n^{O(1)}$, and word size $O(d)$.
\end{theorem}


As before the algorithm is also presented as a public-coin cell-probing scheme, and is transformed into a standard randomized cell-probing scheme by Proposition~\ref{propo-public-vs-private-ANNS}.

This more sophisticated algorithm reuses several components of the simple algorithm in Theorem~\ref{theorem-simple-upper-bound}. For $0\le i\le \lceil{\log_\alpha d}\rceil$, the sets $B_i$ and $C_i$, and the random matrices $M_i$ are constructed in the same way as before. The degenerate cases when $B_0$ or $B_1$ is not empty are also handled in the same as before, so we proceed by assuming Assumption~\ref{assume-0}.

Set $s\triangleq\left(\frac{1}{4}-\frac{1}{2c}\right)k-\frac{1}{4}>1$. We assume that $k =o( \ln\ln d)$, because for some sufficiently large $k=O(\ln\ln d/ \ln\ln\ln d)$, it can be verified that our algorithm already makes $O(1)$ cell-probes per round on average, so there is no need to consider larger number of rounds after that. Hence, we have $1<s<\ln\ln d<\ln\ln n$.

Let $N_j$ be the random $(\frac{c_2}{s}\log n)\times d$ matrices sampled independently as in Definition~\ref{definition-upper-bound-common} and $D_{i,j}\subseteq B$ the subsets of database points constructed from $M_i$ and $N_j$ as in Definition~\ref{definition-upper-bound-common}.  Now the public randomness shared between the cell-probing algorithm and the table are the random matrices $M_i$ and $N_j$ for $0\le j\le i\le \lceil{\log_\alpha d}\rceil$. We make another assumption.

\begin{assumption}\label{assume-2}
For all $0\le j\le i\le \lceil\log_\alpha d\rceil$, 
at most a fraction $n^{-1/s}$ of $B_j$ is not in $D_{i,j}$ and that at most a fraction $n^{-1/s}$ of $C_i$ $\backslash$ $B_{j+1}$ is in $D_{i,j}$.
\end{assumption}
By Lemma~\ref{lemma-upper-bound-assumption}, the error probability of an algorithm that succeeds by assuming both Assumption~\ref{assume-1} and Assumption~\ref{assume-2} is at most $1/4$.

\ifabs{\paragraph{Table construction}}{\paragraph{Table construction.}}
We reuse the $\lceil\log_\alpha d\rceil +1$ tables $T_0,\ldots, T_{\lceil\log_\alpha d\rceil}$ constructed in Theorem~\ref{theorem-simple-upper-bound}. 
In addition, we further construct $(\lceil\log_\alpha d\rceil +1)\times 2^{c_1\log n}$ \emph{auxiliary tables} $\widetilde{T}_{i,j}$ for $0\le i\le \lceil\log_\alpha d\rceil$ and $j\in\{0,1\}^{c_1\log n}$. 
The address of each cell in an auxiliary table $\widetilde{T}_{i,j}$ corresponds to a concatenation $\overline{w}=\langle l,u,w_0,w_1,\ldots,w_s\rangle$ of:
\begin{itemize}
\item a pair of lower and upper thresholds $0\le l\le u\le \lceil\log_\alpha d\rceil$;
\item a special index $1 \le w_0\le s$;
\item $s$ short strings $w_1,\ldots,w_s\in\{0,1\}^{\frac{c_2}{s}\log n}$.  
\end{itemize}
Altogether these correspond to at most $(\log_\alpha d) s2^{c_2\log n}=\mathrm{poly}(n)$ cells in each auxiliary table.
The total number of cells in all tables remains to be a polynomial of $n$.

For $0\le i\le \lceil\log_\alpha d\rceil$, $j\in\{0,1\}^{c_1\log n}$, and any address $\overline{w}=\langle l,u,w_0,w_1,\ldots,w_s\rangle$ of cells in auxiliary table $\widetilde{T}_{i,j}$, the content of the cell $\widetilde{T}_{i,j}[\overline{w}]$ is given as follows: 
For $1\le r\le s$, define $\rho(r)\triangleq\lfloor l+\frac{r-1}{s-1}(u-l)\rfloor$. Let $j=M_ix$ and $w_r=N_{\rho(r)}x$ for $1\le r\le s$. We construct the sets $C_i,D_{i,\rho(1)},\ldots,D_{i,\rho(s)}$ since we now have complete information about the sets.
\begin{itemize}
\item
If there exists an $1\le r\le w_0$ such that $|D_{i,\rho(r)}|>n^{-1/s}|C_i|$, then the cell $\widetilde{T}_{i,j}[\overline{w}]$ stores the smallest such $r$.
\item
If otherwise there is no such $r$, the cell $\widetilde{T}_{i,j}[\overline{w}]$ stores $s+1$.
\end{itemize}

\ifabs{\paragraph{Cell-probing algorithm}}{\paragraph{Cell-probing algorithm.}}
Set $\tau=c'(\frac{1}{k}\log d)^{c/k}$ for some constant $c'\ge\log_\alpha 4$ so that
\[
\left(\frac{\tau}{2}\right)^{\frac{k-1}{2}-2s}\geq \left\lceil\frac{\log_\alpha d}{k}\right\rceil.
\]
The cell-probing algorithm contains at most $(k-1)/2$ \emph{shrinking phases}, succeeded by one final \emph{completion round}. Each shrinking phase contains at most two {rounds}. In every shrinking phase the algorithm makes at most $\frac{\tau-1}{s}+2$ cell-probes to the table, and in the completion round it makes at most $\max\{3\tau,k\}$ parallel cell-probes. Thus the total number of cell-probes is at most 
\ifabs{
\begin{align}
 &\frac{k-1}{2}\left(\left\lceil\frac{\tau-1}{s}\right\rceil+2\right)+\max(3\tau,k)\nonumber\\
 =&O\left(k+\left(\frac{1}{k}\log_\alpha d\right)^{c/k}\right).\label{eq:algorithm-2-cell-probe}
\end{align}
}{
\begin{align}
 \frac{k-1}{2}\left(\left\lceil\frac{\tau-1}{s}\right\rceil+2\right)+\max(3\tau,k)=O\left(k+\left(\frac{1}{k}\log_\alpha d\right)^{c/k}\right).\label{eq:algorithm-2-cell-probe}
\end{align}
}

The algorithm maintains two integers $l$ and $u$, initially $l=0$ and $u=\lceil\log_\alpha d\rceil$. At each shrinking phase $l$ and $u$ are updated, satisfying the invariant that $l<u$, $C_l=\emptyset$ and $C_u\neq\emptyset$. This invariant is satisfied initially since we have $C_0\subseteq B_1=\emptyset$ and $C_{\lceil \log_\alpha d\rceil}\supseteq B_{\lceil \log_\alpha d\rceil}=B$.

The aim of the algorithm is at each shrinking phase to \emph{shrink} the gap $u-l$ by a factor of $O(\tau)$ or to \emph{shrink} the size of $C_u$. When the gap $u-l$ drops below $\max\{3\tau,k\}$ the algorithm enters the completion round, where sets $C_l,\ldots,C_u$ are searched simultaneously by at most $\max\{3\tau,k\}$ parallel cell-probes in one round. We claim that at each shrinking phase, the algorithm updates $l$ and/or $u$ in such a way that either $u'-l'\leq (u-l)/\tau+3$ or $|C_{u'}|\le n^{-1/2s}|C_u|$, where $l'$ and $u'$ denote the updated values of $l$ and $u$, respectively.

\ifabs{\begin{algorithm}[!ht]}
	{\begin{algorithm}}
\ifabs{\caption{$k$-round cell-probing algorithm for large $k$}\label{alg:general}
}{\caption{$k$-round cell-probing algorithm for $\ANNS$ for large $k$}\label{alg:general}}
\begin{algorithmic}
\State Set $\tau\gets(\frac{1}{k}\log d)^{c/k}\log_{\alpha}4$;
\State initialize $l\gets 0$ and $u\gets \log_\alpha d$;
\While {$u-l\geq \max\{3\tau,k\}$}\ifabs{}{\Comment{ shrinking phases}}
\State for $r\in[\tau]$, let $\rho(r)\gets\lfloor l+\frac{r}{\tau}(u-l)\rfloor$;
\For {$j=1$ to $\lceil (\tau-1)/s\rceil$} \ifabs{}{\Comment{computing addresses}}
\State $l_j\gets\rho(1+(j-1)s)$ and $u_j\gets\rho(js)$;
\If {$j=\lceil (\tau-1)/s\rceil$ and $s\nmid\tau-1$}
\State $w_0^j\gets \tau-1-s\lfloor\frac{\tau-1}{s}\rfloor$;
\Else
\State $w_0^j\gets s$;
\EndIf
\State $w_q^j\gets N_{\rho(1+(j-1)s+q-1)}x$ for $1\le q\le w_0^j$;
\EndFor
\State retrieve $T_u[M_ux]$;
\State retrieve $\widetilde{T}_{u,M_ux}[\overline{w}^j]$ for $1\leq j\leq \lceil (\tau-1)/s\rceil$;\ifabs{}{\Comment{$1^\mathrm{st}$ round in shrinking phase}}
\If {$\exists$ $1\leq j\leq \lceil (\tau-1)/s\rceil$ s.t. $\widetilde{T}_{u,M_ux}[\overline{w}^j]\neq s+1$ }
\State let $j^*$ be the smallest such $j$;
\State $r^*\gets (j^*-1)s+\widetilde{T}_{u,M_ux}[\overline{w}^{j^*}]$;
\Else
\State $r^*\gets \tau$;
\EndIf
\If {$r^*=1$}
\State update $u\gets\rho(1)+1$;
\Else
\State retrieve $T_{\rho(r^*-1)-1}[M_{\rho(r^*-1)-1}]x$;\ifabs{}{\Comment{$2^\mathrm{nd}$ round in shrinking phase}}
\If {$T_{\rho(r^*-1)-1}[M_{\rho(r^*-1)-1}]x=\mathrm{EMPTY}$}
\State update $l\gets \rho(r^*-1)-1$;
\If {$r^*<\tau$}
\State update $u\gets\rho(r^*)+1$;
\EndIf
\Else
\State update $u\gets\rho(r^*-1)-1$;
\EndIf
\EndIf
\EndWhile
\State retrieve $T_i[M_ix]$ for all $l+1\leq i\leq u$; \ifabs{}{\Comment{completion round}}
\State $i^*=\min\{l+1\leq i\leq u: T_i[M_ix]\neq \mathrm{EMPTY}\}$;
\State \Return $T_{i^*}[M_{i^*}x]$.
\end{algorithmic}
\end{algorithm}


For $0\le r\le \tau-1$, we denote $\rho(r)\triangleq\lfloor l+\frac{r}{\tau}(u-l)\rfloor$. The cell-probing algorithm proceeds as follows:
\begin{enumerate}
\item
In shrinking phase: among sets $D_{u,\rho(1)},\ldots,D_{u,\rho(\tau-1)}$, the algorithm will first find the smallest $r$ such that $|D_{u,\rho(r)}|>n^{-1/s}|C_u|$. To find such $r$, the algorithm first arranges these sets into $\lceil(\tau-1)/s\rceil$ groups where each group contains up to at most $s$ sets, with each group consumes one parallel cell-probe as follows: for every $1\leq j\leq \lceil(\tau-1)/s\rceil$, let the concatenation $\overline{w}^j=\langle l_j,u_j,w_0^j,w_1^j,\ldots,w_s^j\rangle$ be constructed as:
\begin{itemize}
\item the lower and upper thresholds for the current group: $l_j=\rho(1+(j-1)s)$ and $u_j=\rho(js)$;
\item $w_0^j$ gives the number of sets $D_{u,\rho(r)}$ in the current group: normally it is just $s$ except for the last group due to the rounding, so $w_0^j=\tau-1-s\left\lfloor\frac{\tau-1}{s}\right\rfloor$ if $j=\lceil(\tau-1)/s\rceil$ and $s\nmid\tau-1$, and $w_0^j=s$ if otherwise;
\item for $q=1,2,\ldots, w_0^j$, let $w_q^j=N_{\rho(1+(j-1)s+q-1)}x$.
\end{itemize}
The algorithm reads the contents of cells $T_u[M_ux]$ and $\widetilde{T}_{u,M_ux}[\overline{w}^j]$ for all $1\le j\le \lceil(\tau-1)/s\rceil$ in parallel. 
Let $j^*$ be the smallest $j$ such that $\widetilde{T}_{u,M_ux}[\overline{w}^j]\neq s+1$, or $j^*=(\tau-1)/s+1$ if no such $j$ exists. If $j^*=(\tau-1)/s+1$, let $r^*=\tau$. Otherwise let $r^*=(j^*-1)s+\widetilde{T}_{u,M_ux}[\overline{w}^{j^*}]$. Remember that $\widetilde{T}_{u,M_ux}[\overline{w}^j]= s+1$ means that all $|D_{u,l_j}|,\ldots,|D_{u,\rho(1+(j-1)s+w_0-1)}|\leq n^{-1/s}|C_u|$. Hence such $r^*$ is the smallest $r\in[\tau]$ such that $|D_{u,\rho(r)}|>n^{-1/s}|C_u|$, or $r^*=\tau$ if no such $r$ exists. There are three cases:
\begin{enumerate}
\item 
If $r^*=1$ (CASE 1), the algorithm updates $u$ to $\rho(1) +1$, leaving $l$ unchanged, skips the second round and moves to the next phase. 
\item
Otherwise, the algorithm reads the content of cell $T_{\rho(r^*-1)-1}[M_{\rho(r^*-1)-1}x]$. If the cell is EMPTY (CASE 2), it updates $l$ to $\rho(r^*-1)-1$ and if further $r^*<\tau$, updates $u$ to $\rho(r^*)+1$. 
\item
If $T_{\rho(r^*-1)-1}[M_{\rho(r^*-1)-1}x]\neq \mathrm{EMPTY}$ (CASE 3), the algorithm updates $u$ to $\rho(r^*-1)-1$, leaving $l$ unchanged.
\end{enumerate}
\item
Once the gap $u-l$ drops below $\max\{3\tau,k\}$, the algorithm enters the completion round: it reads the cells $T_i[M_ix]$ for all $l+1\leq i\leq u$ in parallel, finds the smallest $i$ such that $T_i[M_ix]\neq\mathrm{EMPTY}$, and outputs the point stored in that cell. Such $i$ must exist since we know $C_u\neq\emptyset$. Note that the output point is from a nonempty $C_{i}$ such that $C_{i-1}=\emptyset$. With Assumption~\ref{assume-1}, it must be a $\gamma$-approximate nearest neighbor of $x$. 
\end{enumerate}
The pseudocode of the cell-probing algorithm is given in Algorithm~\ref{alg:general}.

We now verify that at each time when the $l$ and $u$ are updated, the invariant that $l<u$, $C_l=\emptyset$ and $C_u\neq\emptyset$ is satisfied. First, in all three cases $l<u$ is obviously satisfied after update. 
\begin{itemize}
\item
Since in CASE 1 and CASE 3 the lower threshold $l$ is not changed, $C_l$ stays empty. And in CASE 2, $T_{\rho(r^*-1)-1}[M_{\rho(r^*-1)-1}x]=\mathrm{EMPTY}$ implies that the set $C_l=C_{\rho(r^*-1)-1}$ is empty.  
\item
In CASE 3, since $T_{\rho(r^*-1)-1}[M_{\rho(r^*-1)-1}x]\neq\mathrm{EMPTY}$, the set $C_u=C_{\rho(r^*-1)-1}$ is nonempty. In CASE 2 when $r^*=\tau$ the upper threshold $u$ is not changed so that $C_u$ stays nonempty. For the remaining cases, since $|D_{u,\rho(r^*)}|>n^{-1/s}|C_s|$, by Assumption~\ref{assume-2}, the set $D_{u,\rho(r^*)}$ must contains at least one point from $B_{\rho(r^*)+1}$. Since $B_{\rho(r^*)+1}\subseteq C_{\rho(r^*)+1}$, the set $C_u=C_{\rho(r^*)+1}$ is nonempty.
\end{itemize}
And note that in CASE 1 and CASE 2, the gap between the updated values of $l$ and $u$ is at most $(\lfloor l+\frac{r^*}{\tau}(l-u)\rfloor +1)-(\lfloor l+\frac{r^*-1}{\tau}(l-u)\rfloor -1)\leq \frac{l-u}{\tau}+3$, and in CASE 3, the size of the new $C_u$ is $|C_{\rho_{r^*-1}-1}| \leq |B_{\rho_{r^*-1}}| \leq |D_{u,\rho_{r^*-1}}|/(1-n^{-1/s}) \leq 2|D_{u,\rho_{r^*-1}}| \leq 2n^{-1/s}|C_u|$.
Therefore, in each shrinking phase, either $u'-l'\leq (u-l)/\tau+3$ or $|C_{u'}|\le n^{-1/2s}|C_u|$, where $l'$ and $u'$ denote the updated values of $l$ and $u$, respectively. 

Notice that as $C_u$ stays nonempty, there are most $2s$ shrinking phases in which $|C_u|$ drops. On the other hand, as long as $u-l\geq \max(3\tau,k)$, we have $(u-l)/\tau+3\leq 2(u-l)/\tau$. Since we choose our $\tau$ to satisfy $\frac{\log_\alpha d}{k}\leq (\tau/2)^{(k-1)/2-2s}$, there can be at most $(k-1)/2-2s$ shrinking phases in which $(u-l)$ shrinks by a factor of $2/\tau$. Hence, overall there can be at most $(k-1)/2$ shrinking phases. Each shrinking phase contains at most 2 rounds, where the algorithm makes $\lceil\frac{\tau-1}{s}\rceil+1$ parallel cell-probes in the first round and one cell-probe in the second round of that phase, and at last in the completion round the algorithm makes $\max(3\tau,k)$ parallel cell-probes. The total number of cell-probes is as given by~\eqref{eq:algorithm-2-cell-probe}.

\ifabs{}{\subsection{A 1-probe protocol for $\lambda$-$\mathsf{ANN}$}}
\ifabs{The algorithms presented in Section~\ref{section-upper-bound} are for the \emph{search} of the \emph{nearest} neighbors. }{The algorithms presented in previous sections are for the \emph{search} of the \emph{nearest} neighbors. }
These highlighted words seem to be critical to this non-trivial cell-probe complexity on a table of polynomial size when both randomization and approximation are allowed. 

Consider a well-known decision version of the problem: the approximate \emph{$\lambda$-near neighbor} problem $\lANN$. Let $\gamma>1,\lambda>0$ be fixed. A point $y\in\{0,1\}^d$ is a $\lambda$-near neighbor to $x\in\{0,1\}^d$ if $\mathrm{dist}(x,y)\le \lambda$. Given a query point $x\in\{0,1\}^d$ and a database $B\subseteq\{0,1\}^d$, the problem $\lANN$ asks to distinguish between the two cases: (1) there is a database point $y\in B$ which is a $\lambda$-near neighbor to $x$, and (2) there is no database point $y\in B$ which is a $\gamma\lambda$-near neighbor to $x$.  For other cases, the answer can be arbitrary. 
This problem has been extensively studied in the context of lower bounds for nearest neighbor search~\cite{borodin1999lower,BarkRaba,liu2004strong, Mihai06,panigrahy2010lower}.



The following is a folklore result: if randomization is allowed then $\lANN$ can be solved with 1-probe on a table of polynomial size. We actually show this for a slightly stronger search problem, the approximate {$\lambda$-near neighbor} search problem $\lANNS$, where if it is the case that there is a $\lambda$-near neighbor in the database, a database point which is a $\gamma\lambda$-near neighbor is output. 


\begin{theorem}\label{theorem-ANN-upper-bound}
Let $\gamma>1$ be any constant. For $n>d$, $\lANNS$ has a randomized cell-probing scheme for with $1$ cell-probe,  table size $n^{O(1)}$, and word size $O(d)$.
\end{theorem}
\begin{proof}
Here we still present a public-coin cell-probing scheme. Apparently the same generic translation in Proposition~\ref{propo-public-vs-private-ANNS} also holds for the $\lANNS$ problem.

Still let $\alpha=\sqrt{\gamma}$.
The table is prepared precisely as in Theorem~\ref{theorem-simple-upper-bound}, with the public random matrices $M_i$ shared between the cell-probing algorithm and the table, and the points from sets $C_i$ which approximate the balls $B_i$ of database points storing in the table.

For the cell-probing algorithm, let $i=\lceil\log_\alpha \lambda\rceil$. Thus $\alpha^i\ge \lambda$ and $\alpha^{i+1}\le \gamma\lambda$. The cell-probing algorithm reads the cell $T_i[M_ix]$ returns the content if it contains a point or returns a NO if it is EMPTY. As argued before, this cell stores a point from $C_i$ if $C_i$ is not empty. Note that if there exist database points which are within distance $\lambda$ from $x$, then $B_i$ is not empty. By Assumption~\ref{assume-1}, $B_i\subseteq C_i\subseteq B_{i+1}$, thus $C_i$ is not empty. In this case a point in $C_i$ must be returned, which is a $\gamma\lambda$-near neighbor to $x$. If no database point is a $\gamma\lambda$-near neighbor to $x$, then $B_{i+1}$ is empty, and due to Assumption~\ref{assume-1}, so is $C_i$, therefore the algorithm may only find $T_i[M_ix]=\mathrm{EMPTY}$ and return with a NO.
\end{proof}

\begin{suppress}
\begin{definition}
	Let $k$ be a positive integer and $\lambda$ a real number with $\lambda\geq1$. We define $\mathcal{V}_\lambda$ to be the distribution of a random $d$-coordinate row verctor in which each coordinate is independently chosen to be $1$ with probability $1/(4\lambda)$ and $0$ otherwise. We define $\mathcal{M}^k_\lambda$ to be the distritubtion of a random $k\times d$ matrix where each row is independently chosen from distribution $\mathcal{V}_\lambda$.
\end{definition}
\begin{lemma}\label{lem:vector equality}
	Let $r\geq1$ and $\gamma>1$. Then, there exist two numbers $\delta_1(\lambda)<\delta_2(\lambda,\gamma)$, both in $[0,1]$, such that $\delta_2(\lambda,\gamma)-\delta_1(\lambda)=(\Delta-\Delta^\gamma)/2$ where $\Delta$ is a function with respect to $\lambda$ such that $1/2\leq\Delta\leq1/\sqrt{e}$, and such that for all $d\geq1$ and for all $x,z\in\{0,1\}^d$,
	\begin{align*}
		\mathrm{dist}(x,z)\leq\lambda &\implies\Pr[Yx\neq Yz]\leq\delta_1(\lambda),\\
		\mathrm{dist}(x,z)>\gamma\lambda &\implies\Pr[Yx\neq Yz]>\delta_2(\lambda,\gamma),
	\end{align*}
	where $Y$ is a random row verctor drawn from distribution $\mathcal{V}_\lambda$.
\end{lemma}
\begin{proof}
	Consider the following equivalent way of choosing $Y$: first choose a set $\mathbf{C}\subseteq[d]$ where each integer in $[d]$ is put in $\mathbf{C}$ independently with probability $1/(2\lambda)$. Then, for each $i\in\mathbf{C}$ let the $i$-th coordinate of $Y$ be chosen uniformly from $\{0,1\}$. For $i\notin\mathbf{C}$, set the $i$-th coordinate of $Y$ to zero. Let $z\in\{0,1\}^d$ be arbitrary and let $h=\mathrm{dist}(x,z)$. If $\mathbf{C}$ does not contain any of the coordinates on which $x$ and $z$ differ, then clearly $Yx=Yz$. This happens with probability $(1-1/(2\lambda))^h$. Otherwise, if $\mathbf{C}$ contains at least one of the coordinates on which $x$ and $y$ differ, the probability that $Yx\neq Yz$ is precisely $1/2$ (note that we perform matrix multiplication on GF(2)). Hence,
	\[
		\Pr[Yx\neq Yz]=\frac{1}{2}(1-(1-\frac{1}{2\lambda})^h).
	\]
	It can be seen that this is a monotonically increasing function of $h$ and that by plugging in $\lambda$ and $\gamma\lambda$ for $h$ one obtains two numbers.
	\par
	Let $\sigma=\delta_2-\delta_1$,
	\begin{align*}
		2\sigma	&=(1-\frac{1}{2\lambda})^\lambda-(1-\frac{1}{2\lambda})^{\gamma\lambda}\\
			   &=((1-\frac{1}{2\lambda})^{2\lambda})^{1/2}-((1-\frac{1}{2\lambda})^{2\lambda})^{\gamma/2}.
	\end{align*}
	Let $\Delta=(1-\frac{1}{2\lambda})^\lambda$. Note that $\Delta^2$ is monotonically increasing when $2\lambda\geq1$, and $\lim_{\lambda\to+\infty}\Delta^2=1/e$. Recall that $\lambda\geq1$. Consequently $1/2\leq\Delta<1/\sqrt{e}$. Hence
	\begin{align*}
		\sigma=(\Delta-\Delta^\gamma)/2,
	\end{align*}
	where $\Delta$ is a function with respect to $\lambda$ such that $1/2\leq\Delta\leq1/\sqrt{e}$.
\end{proof}
\begin{lemma}\label{lem:threshold}
	Let $\lambda\geq1$ and $\gamma>1$. Define $\delta(\lambda,\gamma)\triangleq(\delta_1(\lambda)+\delta_2(\lambda,\gamma))/2$, where $\delta_1,\delta_2$ are as in Lemma \ref{lem:vector equality}. Then, for all $d\geq1$, all $u,v\in\{0,1\}^d$, and all $k\geq1$,
	\begin{align*}
		\mathrm{dist}(u,v)\leq\lambda &\implies\Pr[\mathrm{dist}(Mu,Mv)>\delta(\lambda,\gamma)\cdot k]\leq\exp(-(\delta_2-\delta_1)^2k/2),\\
		\mathrm{dist}(u,v)>\gamma\lambda &\implies\Pr[\mathrm{dist}(Mu,Mv)\leq\delta(\lambda,\gamma)\cdot k]\leq\exp(-(\delta_2-\delta_1)^2k/2),
	\end{align*}
	where $M$ is a random matrix drawn from distribution $\mathcal{M}_\lambda^k$.
\end{lemma}
\begin{proof}
	The lemma follows by combining Lemma \ref{lem:vector equality} with the Chernoff bound: For a sequence of $m$ independent random variables on $\{0,1\}$ such that for all $i$, $\Pr[X_i=1]=p$ for some $p$, $\Pr[\sum X_i>(p+\tau)m]\leq\exp(-2m\tau^2)$ and similarly $\Pr[\sum X_i<(p-\tau)m]\leq\exp(-2m\tau^2)$.
	\par
	For $\mathrm{dist}(u,v)\leq\lambda$,
	\begin{align*}
		&\Pr[\mathrm{dist}(Mu,Mv)>\frac{1}{2}(\delta_1+\delta_2)k]\\
		=&\Pr[\mathrm{dist}(Mu,Mv)>(\delta_1+(\delta_2-\delta_1)/2)k]\\
		\leq&\Pr[\mathrm{dist}(Mu,Mv)>(p+(\delta_2-\delta_1)/2)k]	&(p\leq\delta_1)\\
		\leq&\exp(-(\delta_2-\delta_1)^2k/2).
	\end{align*}
	\par
	Similarly for $\mathrm{dist}(u,v)>\gamma\lambda$,
	\begin{align*}
		&\Pr[\mathrm{dist}(Mu,Mv)\leq\frac{1}{2}(\delta_1+\delta_2)k]\\
		=&\Pr[\mathrm{dist}(Mu,Mv)\leq(\delta_2-(\delta_2-\delta_1)/2)k]\\
		\leq&\Pr[\mathrm{dist}(Mu,Mv)\leq(p-(\delta_2-\delta_1)/2)k]	&(p>\delta_1)\\
		\leq&\exp(-(\delta_2-\delta_1)^2k/2).
	\end{align*}
\end{proof}
\end{suppress}

\begin{suppress}
\todo{This part is a mess, need some serious editing ...}
\begin{definition}[$\lambda$ Aprroximate Near Neighbor]
	For integer $d,n,\lambda\geq1$ and a real number $\gamma>1$, we define the $\lambda$ Approximate Near Neighbor problem $\lANN$ as the data structure query problem given by
	\[
		\mathscr{A}=\{0,1\}^d,\mathscr{B}=\binom{\{0,1\}^d}{n},\mathscr{C}=\{0,1\}
	\]
	\begin{align*}
		\rho=\{(x,y,z)\in\mathscr{A}\times\mathscr{B}\times\mathscr{C}:\\
		(\nexists z'\in y(\mathrm{dist}(x,z')\leq\lambda)\lor z=0)\land(\exists z'\in y(\mathrm{dist}(x,z')\leq\gamma\lambda)\lor z=1)\},
	\end{align*}
	where ``dist" denotes Hamming distance in $\{0,1\}^d$.
\end{definition}

\begin{theorem}[Cell Probe Alogrithm for $\lambda\mathsf{ANN}$]
	Let $\gamma>1$. Then, for $d\ll n\ll 2^d$, $\lANN$ has a cell probe algorithm with $1$ probe, table size $n^{O(1/\sigma^2)}$ and word size $1$, where $\sigma=\Delta-\Delta^\gamma$, $\Delta$ is a function with respect to $\lambda$ such that $1/2\leq\Delta\leq1/\sqrt{e}$.
\end{theorem}
\begin{remark}
	This result is mentioned in a paragraph after Theorem 1.3 in the complete version of \cite{CharkReg}.
\end{remark}
\begin{proof}
	The randomness of our algorithm is public. In another word, we drawn a $M$ from $\mathcal{M}_\lambda^k$, then use the $M$ to build the table and decide the cell to be probed.
	\par
	Let $k=2\frac{\ln(6n)}{(\delta_2-\delta_1)^2}$, Our algorithm runs as follows.
	\par
	In preprocessing phase, for each cell $T[i]$ in the table, we label the cell with a distinct $w_i\in\{0,1\}^k$. If there exists a $v\in y$ such that $\mathrm{dist}(w_i,Mv)\leq\delta k$, we set the content of the cell $T[i]$ as $0$. Otherwise we set the content of the cell $T[i]$ as $1$. Clearly, we need at most $2^k=n^{O(1/(\delta_2-\delta_1)^2)}=n^{O(1/\sigma^2)}$ cells.
	\par
	In querying phase, we probe the cell which is labeled with $Mx$ and return the content directly.
	\par
	Our algorithm could fail only when one of these two bad events happen:
	\begin{enumerate}
		\item There exists a $v\in y$ such that $\mathrm{dist}(x,v)\leq\lambda$ but $\mathrm{dist}(Mx,Mv)>\delta k$.
		\item There exists a $v\in y$ such that $\mathrm{dist}(x,y)>\gamma\lambda$ but $\mathrm{dist}(Mx,Mv)\leq\delta k$.
	\end{enumerate}
	\par
	The former happens with probability
	\begin{align*}
		&\Pr[\text{1. happens}]\\
		\leq&	n\Pr[\text{for some fixed }v\in y \text{ which is satisfying }\mathrm{dist}(x,v)\leq\lambda,\mathrm{dist}(Mx,Mv)>\delta k]	&\text{(Union Bound)}\\
		\leq&	n\exp(-(\delta_2-\delta_2)^2k/2)	&	\text{(Lemma \ref{lem:threshold})}\\
		=&1/6.
	\end{align*}
	\par
	The latter happens with probability
	\begin{align*}
		&\Pr[\text{2. happens}]\\
		\leq&	n\Pr[\text{for some fixed }v\in y \text{ which is satisfying }\mathrm{dist}(x,v)>\gamma\lambda,\mathrm{dist}(Mx,Mv)\leq\delta k]	&\text{(Union Bound)}\\
		\leq&	n\exp(-(\delta_2-\delta_2)^2k/2)	&	\text{(Lemma \ref{lem:threshold})}\\
		=&1/6.
	\end{align*}
	Due to the union bound, our algorithm succeeds with probability as least $2/3$.
\end{proof}
\begin{definition}[$\lambda$ Aprroximate Near Neighbor Searching]
	For integer $d,n,\lambda\geq1$ and a real number $\gamma>1$, we define the $\lambda$ Approximate Near Neighbor Searching problem $\lANNS$ as the data structure query problem given by
	\[
		\mathscr{A}=\{0,1\}^d,\mathscr{B}=\binom{\{0,1\}^d}{n},\mathscr{C}=\{0,1\}^d\cup-1
	\]
	\begin{align*}
		\rho=\{(x,y,z)\in\mathscr{A}\times\mathscr{B}\times\mathscr{C}:\\
		(\nexists z'\in y(\mathrm{dist}(x,z')\leq\lambda)\lor z\in y(\mathrm{dist}(x,z)\leq\gamma\lambda))\land(\exists z'\in y(\mathrm{dist}(x,z')\leq\gamma\lambda)\lor z=-1)\},
	\end{align*}
	where ``dist" denotes Hamming distance in $\{0,1\}^d$.
\end{definition}
\begin{theorem}[Cell Probe Alogrithm for $\lambda\mathsf{ANNS}$]
	Let $\gamma>1$. Then, for $d\ll n\ll 2^d$, $\lANNS$ has a cell probe algorithm with $1$ probe, table size $n^{O(1/\sigma^2)}$ and word size $d+1$, where $\sigma=\Delta-\Delta^\gamma$, $\Delta$ is a function with respect to $\lambda$ such that $1/2\leq\Delta\leq1/\sqrt{e}$.\end{theorem}
\begin{proof}
	The algorithm is similar with the previous one.
	\par
	In preprocessing phase, for each cell $T[i]$ in the table, we label the cell with a distinct $w_i\in\{0,1\}^k$. If there exists a $v\in y$ such that $\mathrm{dist}(w_i,Mv)\leq\delta k$, we set the content of the cell $T[i]$ as arbitrary such $v$. Otherwise we set the content of the cell $T[i]$ as $-1$.
	\par
	In querying phase, we probe the cell which is labeled with $Mx$ and return the content directly.
	\par
	Clearly, the probability our algorithm fails is as low as the previous algorithm.
\end{proof}
\end{suppress}

\begin{suppress}

\begin{theorem}[$k$-adaptive Cell Probe Algorithm for $\ANNS$]
	Let $\gamma>1$ be any constant and k be any integer. Then for $n\geq d$, $\ANNS$ has a trivial $k$-adaptive cell probe algorithm with $O(k(\log d)^{1/k})$ probes, table size $n^{O(1)}$, word size $O(d)$ and published bits at most $O(d+dn)$. Let constant $0<c<1/4$ be some sufficient constant, and $c'>2$ be chosen so that $c'\times ((k-1)/2-ck)/k\geq 1$. When k is super constant, $\ANNS$ has a $k$-adaptive cell probe algorithm with $O(k+(\log d/k)^{c'/k})$ probes, table size $n^{O(1)}$, word size $O(d)$ and published bits at most $O(d+dn)$.
\end{theorem}
\begin{proof}
	Without lost of generality, assume that $\gamma<4$ and let $\alpha\triangleq\sqrt{\gamma}$. Let $x\in\{0,1\}^d$ denote the query point and $B\subseteq\{0,1\}^d$ denote the database. For $i\in\{0,1,\cdots,\log_\alpha d\}$, let $B_i$ be the set of all database points within distance $\alpha^i$ of $x$.
	\par
	We start by checking for the degenerate case in which $x\in B$. This can be done with a constant number of cell probes using the technique of perfect hashing \cite{FKS}. If indeed $x\in B$, the algorithm returns $x$ and ends. Similarly, checking if there exists a point in $B$ within distance 1 of $x$ can also be done with $O(1)$ cell probes by perfect hashing of all the points within distance 1 of $B$ (there are at most $dn$ such points). Again, if such a point is found, the algorithm outputs it and ends. Since it only costs $O(1)$ cell probes to check whether $B_0$ or $B_1$ is empty, we can add it to our first query phase which will be discussed in detail later.
	\par
	Our cell probe algorithms are composed of a \emph{preprocessing phase}, followed by several \emph{query phases}. We will find an $i$ such that $B_i$ is empty but $B_i+2$ is not and will output a point in $B_i+2$; such a point is clearly an $\gamma=\alpha^2$-approximate nearest neighbor of $x$.
	\par
	 Notice that the total input to the $\ANNS$ problem is $d+dn$ bits long. According to the private versus public coin theorem of Newman[14], it is enough to choose the public random string uniformly from a set of at most $O(d+dn)$ special strings. Thus $O(d+dn)$ published bits are enough for our algorithms.
	\par
	Set $s=ck$. We choose public independent random matrices $M_i$ and $N_i$ from distribution $\mathcal{M}_{\alpha^i}^{c_1\log n}$ and $\mathcal{N}_{\alpha^i}^{(c_2\log n)/s}$ respectively. The constant $c_1$ and $c_2$ will be specified later. For $0\leq i\leq\log_\alpha d$ we define the sets
	\[
		C_i\triangleq\{z\in B|\mathrm{dist}(M_ix,M_iz)\leq\delta(\alpha^i,\alpha)\cdot c_1\log n\},
	\]
	\[
		D_{i,j}\triangleq\{z\in C_i|\mathrm{dist}(N_jx,N_jz)\leq\delta(\alpha^i,\alpha)\cdot (c_2\log n)/s\},
	\]
	where $\delta$ is as in Lemma \ref{lem:threshold}.
	\par
	Lemma \ref{lem:threshold} says that $C_i$ is an approximation to $B_i$ in the following sense: a point in $B_i$ may be left out of $C_i$ ( and a point not in $B_{i+1}$ may get into $C_i$) with probability at most $n^{-2}$, if we choose $c_1$ large enough. Similarly, $D_{i,j}$ is an approximation to the set of points in $C_i$ that are within $\alpha^j$ of x.
	\par
	We will make two assumptions. We first \emph{assume} that $B_i\subseteq C_i\subseteq B_{i+1}$ for all $i$. Taking the union bound over all $i$ and all $n$ database points, we see that the assumption is false with probability at most $(\log_\alpha d)\cdot n\cdot n^{-2}\leq\frac{1}{8}$.  Under this assumption, by Lemma \ref{lem:threshold}, a point in $B_j$ is left out of  $D_{i,j}$ ( and a point in $C_i$ $\backslash$ $B_{j+1}$ may get into $D_{i,j}$) with probability at most $n^{-2/s}$ provided we choose $c_2$ large enough. We additionally \emph{assume} that at most a fraction $n^{-1/s}$ of $B_j$ is not in $D_{i,j}$ and that at most a fraction $n^{-1/s}$ of $C_i$ $\backslash$ $B_{j+1}$ is in $D_{i,j}$. Using Markov's inequality followed by a union bound over all $i,j$ we show that the second assumption is false with probability at most $n^{-1/s}\cdot(\log_\alpha d)^2\cdot2\leq\frac{1}{8}$. An assumption is false with probability at most $\frac{1}{4}$.
	\par
	 By the first assumption, we need only find an i such that $C_i$ is empty but $C_i+1$ is not. Since $B_i\subseteq C_i$ is empty and $B_i+2\supseteq C_i+1$ is nonempty, any point in $C_i+1$ solves $\ANNS$.
	 \par
	We first describe the simpler algorithm with $O(k(\log d)^{1/k})$ probes, table size $n^{O(1)}$, word size $O(d)$ and published bits at most $O(d+dn)$. It only needs matrices $M_i$ and sets $C_i$, and we have our first assumption which bounds the error probability of our algorithm.
	\par
	Set $r=c_r(\log d)^{1/k}$, with $c_r$ chosen so that
	\[
	        (r/2)^{k-1}\times r \geq \log_\alpha d
	\]
	\par
	In preprocessing phase, we construct $\log_\alpha d +1$ tables from $T_0$ to $T_{\log_\alpha d}$. There are $2^{c_1\log n}$ cells in each table so that there are at most $(\log_\alpha d+1)2^{c_1\log n}=O((\log_\alpha d)n^{O(1)})=n^{O(1)}$ cells in total. For each cell $T_i(j)$ in $T_i$, if there exists $v\in B$ such that $\mathrm{dist}(j,M_iv)\leq \delta(\alpha^i,\alpha)c_1\log n$, we set the content of the cell as any such $v$. Otherwise we set the content of the cell as EMPTY which implies there does not exist such a near neighbor.
	\par
	In query phases, we do at most $(k-1)$ \emph{shrinking phases} and a \emph{completion phase}. We maintain two integer $l$ and $u$, initialized to $0$ and $\log_\alpha d$ respectively. At the start of each query phase, the algorithm maintains the invariant that $l<u$, $C_l$ is empty and $C_u$ is nonempty. Note that this holds at the very beginning($C_0$ is empty since it is contained in $B_1$). We claim that each of the shrinking phases updates $l$ and $u$ in such a way that $u'-l'\leq (u-l)/r+1$ where $l'$ and $u'$ denote the updated values of $l$ and $u$, respectively. When $u-l$ drops below $r$, the algorithm stops the shrinking phases and moves on to the last completion phase. As long as $u-l\geq r$, $(u-l)/r+1\leq 2(u-l)/r$. Remember that $(r/2)^{k-1}\times r \geq \log_\alpha d$. Hence, there can be at most $k-1$ shrinking phases.
	\par
	Now we describe a shrinking phase. For $q\in[r]$ define $\rho_q\triangleq\lfloor l+\frac{q}{r}(u-l)\rfloor$. We probe the cells $T_{\rho_q}(M_{\rho_q}x)$ for $1\leq q \leq r-1$ in parallel (since matrices $M_i$ are public). Among the $r-1$ cells, let $q'$ be the smallest q such that $T_{\rho_q}(M_{\rho_q}x)$ is nonempty which means that $C_{\rho_q}$ is nonempty, or let $q'=r$ if no such q exists. We update $l$ to $\rho_q'-1$ and $u$ to $\rho_q'$. Note that $C_l$ remains empty and $C_u$ remains nonempty.
	\par
	Once $u-l$ drops below $r$, we enter the completion phase. For $l+1\leq i \leq u$, we probe the cells $T_i(M_ix)$ in parallel. Let $i'$ be the smallest i such that $T_i(M_ix)$ is nonempty. Such $i'$ does exist since we have $T_u(M_ux)$ nonempty. Our algorithm ends up with the cell $T_i'(M_i'x)$ which is an arbitrary point in $C_i'$ such that $C_i'-1$ is empty but $C_i'$ is not. It solves the problem.
	\par
	In our first algorithm, the total number of cell probes is at most
	\[
	        (r-1)\times (k-1)+r+O(1)=O(k(\log_\alpha d)^{1/k})
	\]
	where $O(1)$ is the constant cost in the first query phase to check whether $B_0$ or $B_1$ is empty.
	\par
	Next we describe our second algorithm. When k is any super constant, it solves $\ANNS$ with $O(k+(\log d/k)^{c'/k})$ probes, table size $n^{O(1)}$, word size $O(d)$ and published bits at most $O(d+dn)$. It needs matrices $M_i$ and $N_j$, and sets $C_i$ and $D_{i,j}$. The algorithm is right under both assumptions described above so the error probability is at most $\frac{1}{4}$.
	\par
	Set $r'=c_r'(\log d/k)^{c'/k}$, with $c_r'$ chosen so that
	\[
	        (\frac{r'}{2})^{(k-1)/2-ck}\geq \frac{\log_\alpha d}{k}
	\]
	\par
	In preprocessing phase, we construct $\log_\alpha d +1$ tables from $T_0$ to $T_{\log_\alpha d}$, and  $\log_\alpha d +1$ tables from $S_0$ to $S_{\log_\alpha d}$. Tables $T_i$ are the same as the tables constructed in the first algorithm. For each table $S_i$, we construct $2^{c_1\log n}$ subtables from $S_{i,0}$ to $S_{i,2^{c_1\log n}-1}$, and for each subtable $S_{i,j}$, there are $\frac{r'-1}{s}\times \log_\alpha d\times2^{c_2\log n}$ cells. Thus, there are at most $n^{O(1)}+(\log_\alpha d+1)2^{c_1\log n}(\frac{r'-1}{s}\times \log_\alpha d\times2^{c_2\log n})=n^{O(1)}+O(\log_\alpha d)\times n^{O(1)}=n^{O(1)}$ cells in total. For each cell $S_{i,j}(w)$, we set $l_x=w[0,(\log_\alpha d)-1]$, $u_x=l_x+s-1$ where w[x,y] denotes the concatenation from the xth bit to the yth bit of w. For $q\in[s]$ define $\rho_q\triangleq\lfloor l_x+q-1\rfloor$. Let $M_ux=j$, and for $1\leq q\leq s$ let $N_{u,\rho_q}x=w[\log_\alpha d+(q-1)\log n/s, \log_\alpha d+q\log n/s-1]$. This gives us complete information of $C_u$ as well as $D_{u,\rho_q}$. Construct sets $C_u$,  $D_{u,\rho_1}$, ...,  $D_{u,\rho_s}$. Let $q'$ be the smallest $q\in[s]$ such that $|D_{u,\rho_q}|>n^{-1/s}|C_u|$, or $q'=s+1$ if no such q exists. We set $q'$ as the content of the cell.
	\par
	We do at most $(k-1)/2$ shrinking phases a completion phase. Each of the shrinking phases costs 2-adaptive and the completion phase costs 1-adaptive. We also maintains two integer $l$ and $u$, initialized to $0$ and $\log_\alpha d$ respectively. At the start of each query phase, the algorithm maintains the invariant that $l<u$, $C_l$ is empty and $C_u$ is nonempty. We claim that each of the shrinking phase updates $l$ and/or $u$ in such a way that either $u'-l'\leq (u-l)/r'+3$ or $|C_u'|\leq 2n^{-1/s}|C_u|$, where $l'$ and $u'$ denote the updated values of $l$ and $u$, respectively. Define $\max=\max\{3r',k\}$, when $u-l$ drops below $\max$, the algorithm stops the shrinking phases and moves on to the last completion phase. Since $C_u$ stays nonempty, there are at most $2s$ shrinking phases in which $|C_u|$ drops by a factor $2n^{-1/s}\leq n^{-1/2s}$. On the other hand, as long as $u-l\geq \max$, $(u-l)/r'+3\leq 2(u-l)/r'$. Since $\frac{\log_\alpha d}{m}\leq \frac{\log_\alpha d}{k}\leq (r'/2)^{(k-1)/2-2s}$, there can be at most $(k-1)/2-2s$ shrinking phases in which $(u-l)$ shrinks by a factor of $2/r'$. Thus overall there are at most $(k-1)/2$ shrinking phases.
	\par
	In each of the shrinking phases we do the following: For $q\in[r']$ define $\rho_q\triangleq\lfloor l+\frac{q}{r'}(u-l)\rfloor$. For $1\leq j\leq (r'-1)/s$, compute the concatenation $w_j=(l+(j-1)s)\cdot N_{\rho_{l+(j-1)s}}x\cdots N_{\rho_{l+js-1}}x$, and probe the cells $T_u(M_ux), S_{u,M_ux}(w_j)$ in parallel. Let $j'$ be the smallest j such that $S_{u,M_ux}(w_j)\neq s+1$, or $j'=(r'-1)/s+1$ if no such j exists. If $j'=(r'-1)/s+1$, let $q'=t$. Otherwise let $q'=(j-1)s+S_{u,M_ux}(w_j)$. If $q'=1$(CASE 1), we update $u$ to $\rho_1 +1$, leaving $l$ unchanged, and move to the next phase. Otherwise, we probe the cell $T_{\rho_{q'-1}-1}(M_{\rho_{q'-1}-1}x)$. If it is empty(CASE 2), we update $l$ to $\rho_{q'-1}-1$ and if $q'<t$, update $u$ to $\rho_q'+1$. If it is nonempty(CASE 3), we update $u$ to $\rho_{q'-1}-1$, leaving $l$ unchanged.
	\par
	Once $u-l$ drops below $\max$, we enter the completion phase. It is the same to the completion phase of the first algorithm: for $l+1\leq i\leq u$, we probe the cells $T_i(M_ix)$ in parallel and find an $i$ such that $T_i(M_ix)$ is EMPTY but $T_{i+1}(M_{i+1}x)$ is not. Our algorithm ends up with the content of the cell $T_{i+1}(M_{i+1}x)$.
	\par
	Let us now verify that all the invariants hold after each shrinking phase ends. Clearly, in all three cases, $l<u$. Moreover, in CASE 1 and CASE 3, $C_l$ is empty since $l$ was not changed and in CASE 2, $C_l$ is empty. In CASE 3, $C_u$ is nonempty and in CASE 2 with $q'=t$, $C_u$ is nonempty because $u$ was not changed. In order to show that $C_u$ is nonempty in the remaining cases, recall that by our assumption, $D_{u,\rho_q'}$ contains at most $n^{-1/s}|C_u|$ points from outside $|B_{\rho_q'+1}|$. Therefor, since $|D_{u,\rho_q'}|>n^{-1/s}|C_s|$, it must contain at least one point from $B_{\rho_q'+1}$. In particular, $B_{\rho_j+1}$ and hence $|C_{\rho_q'+1}|$ are nonempty.
	\par
	Notice that in CASE 1 and CASE 2 the difference between the updated values of $l$ and $u$ is at most
	\[
	        (\lfloor l+\frac{q'}{r'}(l-u)\rfloor +1)-(\lfloor l+\frac{q'-1}{r'}(l-u)\rfloor -1)\leq \frac{l-u}{r'}+3
	\]
	and that in CASE 3 the size of the new $C_u$ is
	\[
	        |C_{\rho_{q'-1}-1}| \leq |B_{\rho_{q'-1}}| \leq |D_{u,\rho_{q'-1}}|/(1-n^{-1/s}) \leq 2|D_{u,\rho_{q'-1}}| \leq 2n^{-1/s}|C_u|
	\]
	where we used our assumptions above. Hence, in all three cases each shrinking phase shrinks either $u-l$ or $|C_u|$ as promised.
	\par
	 In our second algorithm, the total number of cell probes is at most
	 \[
	         \frac{k-1}{2}(\frac{r'-1}{s}+2)+\max+O(1)=O(k+(\log_\alpha d/k)^{c'/k})
	 \]
	 where $O(1)$ is the constant cost in the first query phase to check whether $B_0$ or $B_1$ is empty.
\end{proof}

\subsection{Nonadaptive protocol for the Approximate Near Neighbor problem}
\begin{definition}[$\lambda$ Approximate Near Neighbor]
	For integer $d,n,\lambda\geq1$ and a real number $\gamma>1$, we define the $\lambda$ Approximate Near Neighbor problem $\lANN$ as the data structure query problem given by
	\[
		\mathscr{A}=\{0,1\}^d,\mathscr{B}=\binom{\{0,1\}^d}{n},\mathscr{C}=\{0,1\}
	\]
	\begin{align*}
		\rho=\{(x,y,z)\in\mathscr{A}\times\mathscr{B}\times\mathscr{C}:\\
		(\nexists z'\in y(\mathrm{dist}(x,z')\leq\lambda)\lor z=0)\land(\exists z'\in y(\mathrm{dist}(x,z')\leq\gamma\lambda)\lor z=1)\},
	\end{align*}
	where ``dist" denotes Hamming distance in $\{0,1\}^d$.
\end{definition}

\begin{theorem}[Cell Probe Alogrithm for $\lambda\mathsf{ANN}$]
	Let $\gamma>1$. Then, for $d\ll n\ll 2^d$, $\lANN$ has a cell probe algorithm with $1$ probe, table size $n^{O(1/\sigma^2)}$ and word size $1$, where $\sigma=\Delta-\Delta^\gamma$, $\Delta$ is a function with respect to $\lambda$ such that $1/2\leq\Delta\leq1/\sqrt{e}$.
\end{theorem}
\begin{remark}
	This result is mentioned in a paragraph after Theorem 1.3 in the complete version of \cite{CharkReg}.
\end{remark}
\begin{proof}
	The randomness of our algorithm is public. In another word, we drawn a $M$ from $\mathcal{M}_\lambda^k$, then use the $M$ to build the table and decide the cell to be probed.
	\par
	Let $k=2\frac{\ln(6n)}{(\delta_2-\delta_1)^2}$, Our algorithm runs as follows.
	\par
	In preprocessing phase, for each cell $T[i]$ in the table, we label the cell with a distinct $w_i\in\{0,1\}^k$. If there exists a $v\in y$ such that $\mathrm{dist}(w_i,Mv)\leq\delta k$, we set the content of the cell $T[i]$ as $0$. Otherwise we set the content of the cell $T[i]$ as $1$. Clearly, we need at most $2^k=n^{O(1/(\delta_2-\delta_1)^2)}=n^{O(1/\sigma^2)}$ cells.
	\par
	In querying phase, we probe the cell which is labeled with $Mx$ and return the content directly.
	\par
	Our algorithm could fail only when one of these two bad events happen:
	\begin{enumerate}
		\item There exists a $v\in y$ such that $\mathrm{dist}(x,v)\leq\lambda$ but $\mathrm{dist}(Mx,Mv)>\delta k$.
		\item There exists a $v\in y$ such that $\mathrm{dist}(x,y)>\gamma\lambda$ but $\mathrm{dist}(Mx,Mv)\leq\delta k$.
	\end{enumerate}
	\par
	The former happens with probability
	\begin{align*}
		&\Pr[\text{1. happens}]\\
		\leq&	n\Pr[\text{for some fixed }v\in y \text{ which is satisfying }\mathrm{dist}(x,v)\leq\lambda,\mathrm{dist}(Mx,Mv)>\delta k]	&\text{(Union Bound)}\\
		\leq&	n\exp(-(\delta_2-\delta_2)^2k/2)	&	\text{(Lemma \ref{lem:threshold})}\\
		=&1/6.
	\end{align*}
	\par
	The latter happens with probability
	\begin{align*}
		&\Pr[\text{2. happens}]\\
		\leq&	n\Pr[\text{for some fixed }v\in y \text{ which is satisfying }\mathrm{dist}(x,v)>\gamma\lambda,\mathrm{dist}(Mx,Mv)\leq\delta k]	&\text{(Union Bound)}\\
		\leq&	n\exp(-(\delta_2-\delta_2)^2k/2)	&	\text{(Lemma \ref{lem:threshold})}\\
		=&1/6.
	\end{align*}
	Due to the union bound, our algorithm succeeds with probability as least $2/3$.
\end{proof}
\end{suppress}

\section{Lower Bounds}
In this section, we prove the following lower bound for $k$-round randomized approximate nearest neighbor search.
\begin{theorem}[Theorem~\ref{theorem-ANNS-lower-bound-informal}, restated]\label{theorem-ANNS-lower-bound}
For any finite $c_1,c_2>0$, there exists a $c_3>0$ such that the following holds. Let $n,d\geq1$ be sufficiently large integers such that $d\le 2^{\sqrt{\log n}}$ and $n\le 2^{d^{0.99}}$. Let $1\le k\le \frac{\log\log d}{2\log\log\log d}$ be an integer.
If $\ANNS$ has a $k$-round randomized cell-probing scheme with table size $s\leq n^{c_1}$, word size $w\leq d^{c_2}$, such that every query is correctly answered within $t$ total cell-probes in $k$ rounds with probability at least $7/8$, then $t>\frac{c_3}{k}(\log_\gamma d)^{1/k}$.
\end{theorem}

The proof follows the framework given in~\cite{CharkReg}. The framework consists of three main components: 
\begin{enumerate}
\item A reduction from $\LPM$ to $\ANNS$: As observed by Theorem~\ref{theorem-ANN-upper-bound}, it is impossible to prove the lower bound by considering the decision version of $\ANNS$. The longest prefix match problem $\LPM$ captures the nature of $\ANNS$ very well, and meanwhile, is convenient for applying the round eliminations.
\item A round elimination lemma for communication protocols for $\LPM$: Cell-probing schemes are represented as communication protocols. Eliminating a round in any communication protocol for $\LPM$ gives a weaker protocol for the same problem of a smaller scale.
\item Applying the round elimination to $\LPM$ until there is no round left yet the problem is still nontrivial.
\end{enumerate}
Here a simple observation for the $k$-round cell-probing schemes is that $k$ rounds of cell-probes can be simulated by $2k$ rounds of communications. Applying the above framework with this observation, for the first two component, we redo the reduction with a new choice of parameters, and reprove the round elimination lemma for general communication protocols with non-uniform message sizes in different rounds. 

In fact, these variations can be handled routinely by carefully going through the original proofs with new parameters and/or more generic settings.  The most delicate part of our lower bound is our execution of the third step in above framework, which involves an exploitation of the power of round eliminations. This part is in the proof of our main lower bound Theorem~\ref{thm:main lower bound}.

\begin{suppress}

\todo{skip editing below...}
Due to the Theorem \ref{theorem-ANN-upper-bound}, to prove a nontrivial lower bound for $\ANNS$, we must study the difference between $\lANNS$ and $\ANNS$, i.e. we have to focus on the Approximate \emph{Nearest}-Neighbor Search rather than the Approximate \emph{Near}-Neighbor Search. Here comes an idea which exploit another problem, Longest Prefix Match problem, to capture the critical difference between $\lANNS$ and $\ANNS$. Our proof is consisted of 3 main steps.
\begin{enumerate}
	\item \emph{Construct a reduction from $\LPM$ to $\ANNS$.} Which implies a lower bound for $\LPM$ will be a lower bound for $\ANNS$.
	\item \emph{Reduce $k$-round cell-probing scheme to communication protocol.}  A $k$-round cell-probing scheme implies the existence of a communication protocol which solves the problem in $2k$ rounds. It will give a lower bound for $k$-round cell-probing scheme if we refused existence of such a communication protocol.
	\item \emph{Prove that the communication protocol does not exist with the method of Round Elimination.}
\end{enumerate}

\ifabs{\paragraph{Communication Complexity}}{\paragraph{Communication Complexity:}} In the 2-nd and 3-rd steps, we involve the communication model and a series of lower bound prove tools in communication model. Be similar with the data structure query problem, a communication problem can be defined by a relation $\rho\subseteq\mathscr{A}\times\mathscr{B}\times\mathscr{C}$. A communication problem involves two players, Alice and Bob, and a (possibly randomized) communication protocol. At the very beginning, Alice receives input $x\in\mathscr{A}$, Bob receives input $B\in\mathscr{B}$. The two players then exchange messages with each other according to the communication protocol. At the end of the protocol, Alice outputs a $z\in\mathscr{C}$ such that $(x,B,z)\in\rho$. The communication complexity, i.e. the complexity of the communication protocol, is measured by the number of bits the two players send, the number of rounds of message exchange and the probability Alice outputs a correct answer. Note that each message sending is considered as a round of the protocol.

Actually, as you can see in previous sections, our proof is derived from~\cite{CharkReg}. But there are several key barriers in our work which is derived from the differences between the two problems. In our work, we study the problem in limited rounds and give a lower bound concerning the number of cell-probes in total. (Note that we may probe more than one cell in a single round.) This is different from all previous results which probe only one cell in each round and bound the number of rounds. The first barrier in our work is how to reduce a $k$-round cell-probing scheme to a communication protocol. The barrier can be broken by our observation as in Theorem \ref{thm:cell-probe to communication}. The second difficulty is how to deal with the different message sizes in new communication protocol.
The difficulty is overcame by careful analysis in our proof. Another challenge is how to choose the parameters in round elimination procedure to give a good lower bound. Our parameters choosing and the proof is deferred to proof of the Theorem \ref{thm:main lower bound}.

\todo{resume editing from here...}
\end{suppress}

\subsection{Reduction from longest prefix match}
In~\cite{CharkReg}, a reduction from another data structure problem, the longest prefix matching $\LPM$, to $\ANNS$ is constructed.
\begin{definition}[longest prefix match]
For integers $m,n\geq1$ and a finite alphabet $\Sigma$ we define the longest prefix match problem $\LPM$ as the data structure problem that given a query $x\in\Sigma^m$ and a database $B\subseteq\Sigma^m$, $|B|=n$, an answer $z\in B$ must be returned to satisfy that $z$ has the longest common prefix with $x$ among all $y\in B$.
\end{definition}

The reduction in~\cite{CharkReg} maps instances of $\LPM$ to instances of $\ANNS$ without going through the computation model, so it also applies to $k$-round cell-probing schemes.
In order to prove our more refined lower bound, we need to guarantee the same reduction to hold for a more critical parameterization.

Fix the parameters for the problem $\ANNS$. We define $\eta$ and $\beta$ as follows:
\begin{align}
\eta\triangleq1-\frac{\log\log\gamma}{\log\log d}\quad\text{ and }\quad \beta\triangleq1-\frac{c_4}{\log\log d},\label{eq:eta-beta}
\end{align}
where $c_4=2\log201$. Note that it holds that
\begin{align}
\gamma=2^{(\log d)^{1-\eta}}.\label{eq:gamma}
\end{align}


\begin{lemma}[reduction from $\LPM$ to $\ANNS$]\label{lem:reduction}
Let $d$ be a sufficiently large integer, let $\eta$ and $\beta$ be as defined in~\eqref{eq:eta-beta} so that $\gamma$ satisfies~\eqref{eq:gamma}, and set $m\triangleq\lfloor(\log d)^{\eta\beta}\rfloor$.
Let $\Sigma$ be an alphabet of size $\lceil2^{d^{0.99}}\rceil$.
If $\ANNS$ has a $k$-round randomized cell-probing scheme with cell-probe complexity $t$ and success probability $7/8$, using table size $s$ and word size $w$, then so does $\LPM$.
\end{lemma}


Next we explain how to modify the reduction in~\cite{CharkReg} to prove this lemma.

A family of Hamming balls in $\{0,1\}^d$ is said to be \emph{$\gamma$-separated} if the distance between any two points belonging to distinct balls in the family is more than $\gamma$ times the diameter of any ball in the family.
The following lemma is due to Chakrabarti \emph{et al.}~\cite{chakrabarti2003lower}.
\begin{lemma}[rephrased from Lemma 3.2 in~\cite{chakrabarti2003lower}]\label{lemma-separated-family}
Let $d\geq1$ be a large enough integer, and let $\gamma>1$.
Inside a Hamming ball of radius $r$ (where $d^{0.995}\le r\le d$) in $\{0,1\}^d$ there exists a $\gamma$-separated family of $\lceil2^{d^{0.99}}\rceil$ balls, each of radius $r/(8\gamma)$.
\end{lemma}

\begin{lemma}[improved from Lemma 2.3 in~\cite{CharkReg}]\label{lem:existence of space partation tree}
Let $d\geq1$ be a large enough integer, and let $\gamma=2^{(\log d)^{1-\eta}}\ge 3$, as defined in~\eqref{eq:gamma}.
There exists a rooted tree $\mathcal{T}$ whose vertices are Hamming balls in $\{0,1\}^d$ and which satisfies the following properties:
\begin{enumerate}
		\item If $v$ is a child of $u$ in $\mathcal{T}$, then as Hamming balls $v\subset u$.
		\item Each non-leaf vertex of $\mathcal{T}$ has exactly $\lceil2^{d^{0.99}}\rceil$ children.
		\item Each depth-$i$ vertex (the root being a depth-0 vertex) has radius $d/(8\gamma)^i$.
		\item The depth-$i$ vertices form a $\gamma$-separated family of Hamming balls, which means the distance between any two points belonging to distinct balls in the family is more than $\gamma$ times the diameter of any ball in the family.
		\item The leaves of $\mathcal{T}$ are at depth $\lfloor(\log d)^{\eta\beta}\rfloor$, where $\eta$ is as defined in~\eqref{eq:eta-beta}.
	\end{enumerate}
\end{lemma}
\begin{proof}
The proof is almost identical to the proof of Lemma 2.3 in~\cite{CharkReg}, which follows a construction due to Chakrabarti \emph{et al.}~\cite{chakrabarti2003lower}. 
Note that the balls at leaves have radius of at least $d/(8\gamma)^{\lfloor(\log d)^{\eta\beta}\rfloor}$.
By our choices of $\eta$ and $\beta$ as defined in~\eqref{eq:eta-beta}, it can be verify that for large enough $d$,
\[
\frac{d}{(8\gamma)^{\lfloor(\log d)^{\eta\beta}\rfloor}}\geq d^{0.995}.
\]
Then by Lemma~\ref{lemma-separated-family}, we have the suitable tree $\mathcal{T}$ by a natural recursive construction.
\end{proof}

Given this tree $\mathcal{T}$, the reduction from $\LPM$ with the new string length $m=\lfloor(\log d)^{\beta\eta}\rfloor$ to $\ANNS$ can be constructed by reusing the mapping from $\LPM$ instances to $\ANNS$ instances described in the proof of Lemma 2.4 in~\cite{CharkReg} as a blackbox.

\ifabs{
\subsection{Round elimination}
}{
\subsection{Round elimination for communication protocols}}\label{section-round-elimination}
We now consider communication protocols between two players Alice and Bob in Yao's model of communication complexity~\cite{yao1979some}. We refer the readers to the nice textbook by Kushilevitz and Nisan~\cite{kushilevitz2006communication} for formal definitions of various concepts, e.g.~private-coin protocols.

We assume Alice and Bob send messages to each other alternatively. We use two vectors $\mathbf{A}=(a_1,a_2,\ldots,a_k)$ and $\mathbf{B}=(b_1,b_2,\ldots,b_k)$ to respectively denote the lengths of messages sent by Alice and Bob in each round.

\begin{definition}
Let $\mathbf{A}=(a_1,a_2,\ldots,a_k)\in\mathbb{R}^k_{\ge0}$ and $\mathbf{B}=(b_1,b_2,\ldots,b_k)\in\mathbb{R}^k_{\ge0}$. An $\langle\mathbf{A},\mathbf{B},2k\rangle^A$-protocol is a $2k$-round communication protocol, in which Alice and Bob send messages to each other alternatively, with Alice sending the first message, with the size of Alice's $i$-th message being exactly $\lfloor a_i\rfloor$ bits, and the size of Bob's $i$-th message being exactly $\lfloor b_i\rfloor$ bits.
The superscript ``A'' indicates that Alice sends the first message.

For $\mathbf{A}=(a_1,\ldots,a_{k-1})\in\mathbb{R}^{k-1}_{\ge0}$ and $\mathbf{B}=(b_1,\ldots,b_k)\in\mathbb{R}^k_{\ge0}$, we call such a protocol an $\langle\mathbf{A},\mathbf{B},2k-1\rangle^B$-protocol if the first message is sent by Bob.
\end{definition}

A data structure problem $\rho\subseteq\mathscr{A}\times\mathscr{B}\times\mathscr{C}$ is naturally a communication problem: Alice is give a query $x\in\mathscr{A}$ as input, Bob is given a database $B\in\mathscr{B}$ as input, and Alice is asked to output a correct answer $z\in\mathscr{C}$ satisfying $(x,B,z)\in\rho$ after communicating with Bob. As observed in~\cite{miltersen1995data}, any cell-probing scheme is actually a communication protocol, with Alice being the cell-probing algorithm and Bob being the table.

\begin{proposition}\label{thm:cell-probe to communication}
If a data structure problem $\rho$ has a randomized cell-probing scheme using table size $s$ and word size $w$ bits, such that every query is answered correctly within $t$ total cell-probes in $k$ rounds with probability $1-\epsilon$, then $\rho$ has a private-coin $\langle\mathbf{A},\mathbf{B},2k\rangle^A$-protocol with ${a}_i=t_i\lceil\log s\rceil$ and ${b}_i=t_iw$ for every $1\le i\le k$, for some $t_1,t_2,\ldots,t_k\ge 0$ that $t=\sum_{i=1}^kt_i$, such that Alice outputs a correct answer with probability at least $1-\epsilon$.
\end{proposition}
Here the natural interpretation is that each round of $t_i$ many parallel cell-probes can be simulated by two rounds of communications: Alice sends the addresses of the $t_i$ cells, each of $\lceil\log s\rceil$ bits, to Bob, and Bob responds by sending back the contents of these $t_i$ cells, each of $w$ bits.

Let $\mathbf{A}=(a_1,a_2,\ldots,a_k)$ and $\mathbf{B}=(b_1,b_2,\ldots,b_{k'})$ be two vectors, and $c\in\mathbb{R}$ be a number. We introduce some notations:
\begin{itemize}
\item let $c\mathbf{A}=(ca_1, ca_2,\ldots,ca_k)$;
\item denote by $(\mathbf{A},\mathbf{B})$, or simply $\mathbf{A}\mathbf{B}$, the concatenation: $\mathbf{A}\mathbf{B}=(a_1,\ldots,a_k,b_1,\ldots,b_{k'})$;
\item denote by $(c,\mathbf{A})$ the concatenation of $(c)$ and $\mathbf{A}$: $(c,\mathbf{A})=(c, a_1, a_2,\ldots,a_k)$
\item denote by $\mathbf{A}^{i-}$ the suffix of $\mathbf{A}$ starting at position $i$: $\mathbf{A}^{i-}=(a_i,a_{i+1},\ldots, a_{k})$.
\end{itemize}

\begin{suppress}
We denote the size of messages sent by Alice and Bob by two vector $\mathbf{A}$ and $\mathbf{B}$. For the convenience of communication protocol representation, we additionally define several operators on vectors.
\begin{definitions}[Operators on Vectors]
	$\phantom{ }$
	\begin{itemize}
		\item For a vector $\mathbf{A}\in\mathbb{R}^n$, $\mathbf{A}_i\in\mathbb{R}$ denotes the $i$-th entry of $\mathbf{A}$.
		\item For a vector $\mathbf{A}\in\mathbb{R}^n$ and a number $c\in\mathbb{R}$, $c\mathbf{A}\in\mathbb{R}^n$ denotes a new vector such that $\forall i,(c\mathbf{A})_i=c(\mathbf{A}_i)$.
		\item For two vectors $\mathbf{A}\in\mathbb{R}^m$ and $\mathbf{B}\in\mathbb{R}^n$, $\mathbf{A}\cdot \mathbf{B}\in\mathbb{R}^{m+n}$ denotes a vector which is concatenation of $\mathbf{A}$ and $\mathbf{B}$ such that $\forall i\leq n, (\mathbf{A}\cdot \mathbf{B})_i=\mathbf{A}_i$ and $\forall i>n, (\mathbf{A}\cdot \mathbf{B})_i=\mathbf{B}_{i-m}$.
		\item For a number $c\in\mathbb{R}$ and a vector $\mathbf{A}\in\mathbb{R}^n$, $c\cdot \mathbf{A}\in\mathbb{R}^{1+n}$ denotes a vector such that $(c\cdot \mathbf{A})_1=c$ and $\forall i>1, (c\cdot \mathbf{A})_i=\mathbf{A}_{i-1}$.
		\item For a vector $\mathbf{A}\in\mathbb{R}^n$, $\mathrm{sfx}(\mathbf{A},i)\in\mathbb{R}^{n-i+1}$ denotes a vector such that $\forall j, \mathrm{sfx}(\mathbf{A},i)_j=A_{i+j-1}$
	\end{itemize}
\end{definitions}
Now we are prepared to describe the notations for communication protocols and show a reduction from $k$-round cell-probing scheme to nonuniform communication protocol.

\begin{definition}[Notation for Communication Protocols]
	Let $\mathbf{A},\mathbf{B}\in\mathbb{R}^k$, an $\langle\mathbf{A},\mathbf{B},2k\rangle^A$-protocol is a $2k$-round protocol, with the size of Alice's $i$-th message being exactly $\lceil \mathbf{A}_i\rceil$, and the size of Bob's $i$-th message being exactly $\lfloor \mathbf{B}_i\rfloor$. The superscript ``A'' indicates that Alice sends the first message. We use a ``B'' superscript if the first message is sent by Bob.
\end{definition}
\begin{theorem}\label{thm:cell-probe to communication}
If a data structure problem has a private-coin randomized cell-probing scheme with table size $s$, word size $w$ bits such that every query is correctlly answered within $t$ total cell-probes in $k$ rounds with probability at least $1-\epsilon$, then it has a (private-coin) randomized $\langle\mathbf{A},\mathbf{B},2k\rangle^A$-protocol such that $\mathbf{A,B}\in\mathbb{R}^k$, and $\forall i, \mathbf{A}_i=t_i\lceil\log s\rceil,\mathbf{B}_i=t_iw$, and $\sum_{i=1}^kt_i=t$, and Alice outputs correct answer with probability at least $1-\epsilon$.
\end{theorem}
\begin{proof}
	We construct such a protocol according to the $k$-round, it suffices to prove the theorem. The protocol uses the $k$-round cell-probing scheme as a black box. Bob plays the role of table in the cell-probing scheme, Alice plays the role of cell-probing algorithm. At the very beginning, Bob constructs a table according the the code $\mathcal{T}$ in the cell-probing scheme with his input $B$. In $i$-th round of the cell-probing algorithm, Alice use the $i$-th lookup function $L_i$ to identify the addresses of the table cells to be probed in $i$-th round of the cell-probing algorithm, and sends the concatenation of the $t_i$ addresses to Bob. Alice can do this since Bob always sends back the contents of the table cells to Alice. Note that Alice sends $t_i\lceil\log s\rceil$ bits in Alice's $i$-th round since each address of cell can be identified with at most $\lceil\log s\rceil$ bits. Upon Bob receiving Alice's $i$-th message, Bob makes the $t_i$ cell-probes, and sends back the concatenation of the contents of the table cells in the order of Alice sends the addresses. Clearly, Bob sends $t_iw$ bits in his $i$-th round. At the end of the protocol, Alice uses the truth table $A$ to map the contents of all the probed cells to an output $z$. Obviously, the probability the protocol outputs a correct answer is not less than the probability the cell-probe scheme outputs a correct answer.
\end{proof}
\end{suppress}

The following is the round elimination lemma for $\LPM$ that plays a central role in proving the lower bound. The lemma is generalized from a simpler round elimination lemma in~\cite{CharkReg} to adapt to the non-uniform amount of information communicated in each round.

\begin{lemma}[round elimination lemma for $\LPM$]\label{lemma-round-elimination}\label{LEMMA-ROUND-ELIMINATION}
Let $m,n,p,q$ be positive integers such that $2p\mid m$, $q\mid n$, and $n\leq|\Sigma|$. Let $0<\epsilon,\delta<1$ and $\mathbf{A,B}\in\mathbb{R}^k_{\ge 0}$. There is a universal constant $C>0$ such that the followings hold.
Assume that $k\geq1$ and $\frac{2a_1}{p}\geq C$. If $\LPM$ has a private-coin $\langle\mathbf{A},\mathbf{B},2k\rangle^A$-protocol with error probability $\epsilon$, then $\mathsf{LPM}^{\Sigma}_{m/2p, n/q}$ has a private-coin $\langle\mathbf{A}',\mathbf{B}',2k-2\rangle^A$-protocol with error probability $\epsilon'$, where
\begin{align*}
	\ifabs{
\mathbf{A}'&=\left(1+\frac{2a_1}{\delta p a_2}\right)\mathbf{A}^{2-},
\quad
\mathbf{B}'=\mathbf{B}^{2-},\\
\text{and }
\epsilon'&=\epsilon+2\delta+\sqrt{b_12^{\frac{2a_1}{\delta p}}/q}.
	}
	{
\mathbf{A}'=\left(1+\frac{2a_1}{\delta p a_2}\right)\mathbf{A}^{2-},
\quad
\mathbf{B}'=\mathbf{B}^{2-},
\quad
\text{ and }
\epsilon'=\epsilon+2\delta+\sqrt{b_12^{\frac{2a_1}{\delta p}}/q}.
}
\end{align*}
\end{lemma}

\ifabs{
The proof of this lemma is in Appendix~\ref{appendix-round-elimination}.
}{
The rest of Section~\ref{section-round-elimination} is dedicated to the proof of this lemma.
The proof is almost identical to the one in~\cite{CharkReg}, except for the part dealing with non-uniform message sizes. We include the proof here for the completeness of the paper. 

We only need to show that the following two propositions.
\begin{enumerate}
\item[Part I.] Assume that $k\geq1$ and $2a_1/p\geq C$. If $\LPM$ has a private-coin $\langle\mathbf{A},\mathbf{B},2k\rangle^A$-protocol with error probability $\epsilon$, then $\mathsf{LPM}^{\Sigma}_{m/p, n}$ has a private-coin $\langle\mathbf{A}',\mathbf{B}',2k-1\rangle^B$-protocol with error probability $\epsilon+2\delta$, where
\ifabs{
\begin{align}
\mathbf{A}'
=
\left(1+\frac{2a_1}{\delta p a_2}\right)\mathbf{A}^{2-},\,\,
\mathbf{B}'
=
\left(b_12^{\frac{2a_1}{\delta p}}, \mathbf{B}^{2-}\right).\label{eq:round-elimination-1}
\end{align}
}
{
\begin{align}
\mathbf{A}'
=
\left(1+\frac{2a_1}{\delta p a_2}\right)\mathbf{A}^{2-}
\quad\text{ and }\quad
\mathbf{B}'
=
\left(b_12^{\frac{2a_1}{\delta p}}, \mathbf{B}^{2-}\right).\label{eq:round-elimination-1}
\end{align}
}
\item[Part II.] Assume that $n\leq|\Sigma|$. If $\LPM$ has a private-coin $\langle\mathbf{A},(b_0,\mathbf{B}),2k+1\rangle^B$-protocol with error probability $\epsilon$, then $\mathsf{LPM}^{\Sigma}_{m-1, n/q}$ has a private-coin $\langle\mathbf{A},\mathbf{B},2k\rangle^A$-protocol with error probability $\epsilon+\sqrt{b_0/q}$.
\end{enumerate}
The round elimination lemma (Lemma~\ref{lemma-round-elimination}) follows by combining these two propositions together, and weakening the resulting statement from $m/p-1$ to $m/2p$. The proofs of these two propositions will follow the same routine as in~\cite{CharkReg}, with a generalization to deal with non-uniform message sizes.

The following is a typical proposition in the context of round elimination of communication protocols. Here we prove a version which is suitable for our setting.
\begin{lemma}[message switching lemma]\label{lem:message switching}
Let $P$ be a deterministic $\langle\mathbf{A},\mathbf{B},2k\rangle^A$-protocol with $k\geq1$. Then there exists a deterministic $\langle\mathbf{A}',\mathbf{B}',2k-1\rangle^B$-protocol, where $\mathbf{A}'=(1+a_1/a_2)\mathbf{A}^{2-}$ and $\mathbf{B}'=(b_12^{a_1},\mathbf{B}^{2-})$, that computes the exact same problem as $P$.
\end{lemma}
\begin{proof}
There are at most $2^{a_1}$ different messages that Alice may send as the first message.
Bob starts the new protocol by sending his at most $2^{a_1}$ different responses as in $P$.
If $k=1$, the new protocol stops after this. Otherwise, let Alice's first message be the concatenation of her first two messages in $P$. And then the protocol continues just as in $P$. This increases the sizes of Alice's messages (in fact only her first message in the new protocol) by a factor of at most $(1+a_1/a_2)$.
\end{proof}


We need to define some concepts for the information complexity of communications.
Let $P$ be a communication protocol and $\mathcal{D}$ a joint distribution on the possible inputs to Alice and Bob.
Let $\mathrm{err}(P,\mathcal{D})$ denote the probability of $P$ being error under input distribution $\mathcal{D}$.
Let $\mathcal{D}_A$ denote the marginal distribution of $\mathcal{D}$ on Alice's inputs and $\mathcal{D}_B$ the marginal distribution on Bob's inputs.

\begin{definition}[information cost]
The \concept{information cost} of a private-coin protocol $P$ with respect to input distribution $\mathcal{D}$, denoted $\mathrm{icost}(P,\mathcal{D})$, is defined to be the mutual information $I(X:\mathrm{msg}(P,X))$, where $X$ is a random input drawn from $\mathcal{D}_A$ (if Alice starts $P$) or $\mathcal{D}_B$ (if Bob starts $P$), and $\mathrm{msg}(P,x)$ denotes the first message in protocol $P$ if the sender's input is $x$. %
\end{definition}

The next two generic lemmas hold for general communication protocols with non-uniform sizes of messages, which apply to our setting.
\begin{lemma}[uninformative message lemma~\cite{senp03}]\label{lem:uniformative message}
Let $P$ be a private-coin $\langle\mathbf{A},\mathbf{B},2k\rangle^A$-protocol for a communication problem $\rho$. Then for any input distribution $\mathcal{D}$, there is a deterministic $\langle\mathbf{A}^{2-},\mathbf{B},2k-1\rangle^B$-protocol $P'$ for $\rho$ such that $\mathrm{err}(P',\mathcal{D})\leq\mathrm{err}(P,\mathcal{D})+\sqrt{\mathrm{icost}(P,\mathcal{D})}$.
\end{lemma}

\begin{lemma}[message compression lemma~\cite{CharkReg}] \label{lem:message compression}
Let $P$ be a private-coin $\langle\mathbf{A},\mathbf{B},2k\rangle^A$-protocol for a communication problem $\rho$. Then for any input distribution $\mathcal{D}$ and any $a>0$, there is a deterministic $\langle (a,\mathbf{A}^{2-}),\mathbf{B},2k\rangle^A$-protocol $P'$ for $\rho$ such that $\mathrm{err}(P',\mathcal{D})\leq\mathrm{err}(P,\mathcal{D})+(2\cdot\mathrm{icost}(P,\mathcal{D})+C)/a$, where $C>0$ is a universal constant.
\end{lemma}


Now we are ready to show the two propositions that support the round elimination lemma. This is done by going through the same proof in~\cite{CharkReg} with a different parameterization.

\ifabs{\paragraph{Proof of Part I}}{\paragraph{Proof of Part I.}}
Assume that $\LPM$ has private-coin $\langle\mathbf{A},\mathbf{B},2k\rangle^A$-protocol with error probability $\epsilon$ for a $k\geq1$. We then construct a private-coin $\langle\mathbf{A}',\mathbf{B}',2k-1\rangle^B$-protocol with error probability $\epsilon+2\delta$ for $\mathsf{LPM}^{\Sigma}_{m/p, n}$ where $\mathbf{A}'$ and $\mathbf{B}'$ are given in~\eqref{eq:round-elimination-1}, when $2a_1/p\geq C$ for a universal constant $C>0$.


Let $S\triangleq\Sigma^{m/p}$. By Yao's min-max principle, it suffices to give a deterministic protocol for $\mathsf{LPM}^{\Sigma}_{m/p, n}$ with the same message lengths and distributional error on any input distribution $\mathcal{D}$ on $S\times S^n$. Fix an input distribution $\mathcal{D}$ over $S\times S^n$.
Define the following distributions:
	\begin{itemize}
		\item[$\mathcal{I}$:] Let $\mathcal{I}$ defenote the distribution over $[p]\times S^*$ obtained as follows: choose $i\in[p]$ uniformly at random and draw $\sigma\in S^{i-1}$ from $\mathcal{D}_A^{i-1}$. Recall that $\mathcal{D}_A$ denote the marginal distribution on Alice's inputs.
		\item[$\mathcal{D}_{i,\sigma}$:] Let $s$ be some arbitrarily fixed element in $S$. For each pair $(i,\sigma)$ of $\mathcal{I}$ we define a distribution $\mathcal{D}_{i,\sigma}$ on $S^p\times (S^p)^n$ as follows: draw a sample $(x,y)$ from $\mathcal{D}$, independently draw $p-i$ strings $X_{i+1},\cdots,X_p$ from $\mathcal{D}_A$, then output $(\sigma x X_{i+1}\cdots X_p,\sigma ys^{p-i})$. Note that $y$ is a set of string, and the $\sigma ys^{p-i}$ denote the set of strings $\{\sigma\tau s^{p-i}:\tau\in y\}$.
		\item[$\widetilde{\mathcal{D}}$:] Finally, let $\widetilde{\mathcal{D}}$ be the distribution on $S^p\times(S^p)^n$ obtained by drawing a $(i,\sigma)$ from $\mathcal{I}$ and outputting a sample from $\mathcal{D}_{i,\sigma}$.
	\end{itemize}
	By  the easy direction of Yao's min-max principle, there is a deterministic $\langle\mathbf{A},\mathbf{B},2k\rangle^A$-protocol $P$ for $\LPM$ with distributional error at most $\epsilon$ under distribution $\widetilde{\mathcal{D}}$. By definition,
	\begin{align*}
		\mathbf{E}_{i,\sigma}[\mathrm{err}(P,\mathcal{D}_{i,\sigma})]=\mathrm{err}(P,\widetilde{\mathcal{D}})\leq\epsilon,
	\end{align*}
	where the expectation is taken over $(i,\sigma)$ which is sampled from $\mathcal{I}$.

	Let $X=X_1X_2\cdots X_p$ be distributed according to $\widetilde{\mathcal{D}}_A=\mathcal{D}_A^p$. Then by definition,
	\ifabs{
	\begin{align*}
		\mathrm{icost}(P,\widetilde{\mathcal{D}})	&=\mathrm{I}(X:\mathrm{msg}(P,X))\\
		&=\sum_{i\in[p]}\mathrm{I}(X_i:\mathrm{msg}(P,X)|X_1\cdots X_{i-1})\\
		&=\sum_{i\in[p]}\mathbf{E}_\sigma[\mathrm{I}(X_i:\mathrm{msg}(P,X)|X_1\cdots X_{i-1}=\sigma)]\\
		&=p\cdot\mathbf{E}_{i,\sigma}[\mathrm{I}(X_i:\mathrm{msg}(P,X)|X_1\cdots X_{i-1}=\sigma)].
	\end{align*}
	Where the second equality is by the chain rule for mutual information. Note that $\mathrm{icost}(P,\widetilde{\mathcal{D}})\leq a_1$, since Alice's first message is of length $a_1$. Hence
}
{
	\begin{align*}
		\mathrm{icost}(P,\widetilde{\mathcal{D}})	&=\mathrm{I}(X:\mathrm{msg}(P,X))\\
		&=\sum_{i\in[p]}\mathrm{I}(X_i:\mathrm{msg}(P,X)|X_1\cdots X_{i-1})	&	\text{(chian rule)}\\
		&=\sum_{i\in[p]}\mathbf{E}_\sigma[\mathrm{I}(X_i:\mathrm{msg}(P,X)|X_1\cdots X_{i-1}=\sigma)]\\
		&=p\cdot\mathbf{E}_{i,\sigma}[\mathrm{I}(X_i:\mathrm{msg}(P,X)|X_1\cdots X_{i-1}=\sigma)].
	\end{align*}
	Note that $\mathrm{icost}(P,\widetilde{\mathcal{D}})\leq a_1$, since Alice's first message is of length $a_1$. Hence
}
	\begin{align*}
		\mathbf{E}_{i,\sigma}[\mathrm{I}(X_i:\mathrm{msg}(P,X)|X_1\cdots X_{i-1}=\sigma)]\leq\frac{a_1}{p}.
	\end{align*}
	Due to the linearity of expectation,
	\begin{align*}
		&\mathbf{E}_{i,\sigma}\left[\mathrm{err}(P,\mathcal{D}_{i,\sigma})+
		\frac{2\cdot\mathrm{I}(X_i:\mathrm{msg}(P,X)|X_1\cdots X_{i-1}=\sigma)}{2a_1/(\delta p)}\right]\\
		=&\mathbf{E}_{i,\sigma}[\mathrm{err}(P,\mathcal{D}_{i,\sigma})]+\frac{2\mathbf{E}_{i,\sigma}[\mathrm{I}(X_i:\mathrm{msg}(P,X)|X_1\cdots X_{i-1}=\sigma)]}{2a_1/(\delta p)}\\
		\leq&\epsilon+\frac{2a_1/p}{2a_1/(\delta p)}
		=\epsilon+\delta.
	\end{align*}
	By the averaging principle, there is an integer $i\in[p]$ and a string $\sigma\in S^{i-1}$ such that
	\begin{align*}
		\mathrm{err}(P,\mathcal{D}_{i,\sigma})+\frac{2\cdot\mathrm{I}(X_i:\mathrm{msg}(P,X)|X_1\cdots X_{i-1}=\sigma)}{2a_1/(\delta p)}	&\leq\epsilon+\delta.
	\end{align*}
Fix the pair $(i,\sigma)$ to satisfy above. We can now construct a private-coin protocol $Q''$ for $\mathsf{LPM}^\Sigma_{m/p,n}$ which uses $P$ as a black box. It works as follows: given an input $(x,y)\in S\times S^n$, Alice constructs a string $\tilde{x}\triangleq\sigma xX_{i+1}\cdots X_p$ where the $X_j$'s are random strings drawn independently from $\mathcal{D}_A$ using her private-coins, and Bob constructs the set of strings $\tilde{y}\triangleq\sigma ys^{p-i}$. They then run protocol $P$ on input $(\tilde{x},\tilde{y})$ and output the $i$-th block of the output of $P$. Note that $(\tilde{x},\tilde{y})$ is distributed according to $\mathcal{D}_{i,\sigma}$ if $(x,y)$ is distributed according to $\mathcal{D}$. Clearly, due to the definition of $\mathsf{LPM}$,  $Q''$ works as $P$ works. Therefore,
	\[
		\mathrm{err}(Q'',\mathcal{D})\leq\mathrm{err}(P,\mathcal{D}_{i,\sigma}).
	\]
	Moreover,
	\[
		\mathrm{icost}(Q'',\mathcal{D})=\mathrm{I}(X_i:\mathrm{msg}(P,X)|X_1\cdots X_{i-1}=\sigma).
	\]
	Applying the message compression lemma (Lemma~\ref{lem:message compression}) to $Q''$, we have a deterministic $\langle(2a_1/(\delta p),\mathbf{A}^{2-}),\mathbf{B},2k\rangle^A$-protocol $Q'$ for $\mathsf{LPM}^\Sigma_{m/p,n}$, whose error on distribution $\mathcal{D}$ is bounded as
	\begin{align*}
		\mathrm{err}(Q',\mathcal{D})	&\leq\mathrm{err}(Q'',\mathcal{D})+(2\cdot\mathrm{icost}(Q'',\mathcal{D})+C)/(2a_1/(\delta p))\\
		&\leq\epsilon+\delta+\delta\frac{C}{2a_1/p}.
	\end{align*}
Recall that we assume $2a_1/p\geq C$. The above error is bounded by $\epsilon+2\delta$.

	Applying the message switching lemma (Lemma~\ref{lem:message switching}) to $Q'$, we have a deterministic $\langle\mathbf{A}',\mathbf{B}',2k-1\rangle^B$-protocol $Q$ with error probability $\epsilon+2\delta$ for $\mathsf{LPM}^{\Sigma}_{m/p, n}$ where $\mathbf{A}'$ and $\mathbf{B}'$ are given in~\eqref{eq:round-elimination-1}. Applying this to every joint distribution $\mathcal{D}$ over the inputs with Yao's min-max lemma, we prove the first proposition supporting the round elimination lemma.

	\ifabs{\paragraph{Proof of Part II}}{\paragraph{Proof of Part II.}}
%

Assume that $n\leq|\Sigma|$ and $\LPM$ has a private-coin $\langle\mathbf{A},(b_0,\mathbf{B}),2k+1\rangle^B$-protocol with error probability $\epsilon$. 
Let $S\triangleq\Sigma^{m-1}$. For an arbitrary input distribution $\mathcal{D}$ on $S\times S^{n/q}$, we show the existence of a deterministic $\langle\mathbf{A},\mathbf{B},2k\rangle^A$-protocol for $\mathsf{LPM}^\Sigma_{m-1,n/q}$ with  error $\epsilon+\sqrt{b_0/q}$ on distribution $\mathcal{D}$. By Yao's min-max principle, this is sufficient.

	Since $q\mid n$ implies that $q\leq n\leq|\Sigma|$, we can fix $q$ distinct strings $s_1,\cdots,s_q\in\Sigma$. We now define two distributions based on $\mathcal{D}$.
	\begin{itemize}
		\item[$\mathcal{D}_i$:] For each $i\in[q]$, let $\mathcal{D}_i$ be the distribution on $\Sigma S\times (\Sigma S)^n$ obtained as follows: independently draw $q$ samples $(x_1,y_1),\cdots,(x_q,y_q)$ from $\mathcal{D}$, and output $(s_ix_i,s_1y_1\cup\cdots\cup s_qy_q)$.
		\item[$\widetilde{\mathcal{D}}$:] We also contruct a distribution $\widetilde{\mathcal{D}}$ on $\Sigma S\times(\Sigma S)^n$ as follows: choose $i\in[q]$ uniformly at random, and output a sample from $\mathcal{D}_i$.
	\end{itemize}
	By the easy direction of Yao's min-max principle there is a deterministic $\langle\mathbf{A},(b_0,\mathbf{B}),2k+1\rangle^B$-protocol $P$ for $\LPM$ with error at most $\epsilon$ on input distribution $\widetilde{\mathcal{D}}$. By definition,
	\[
		\mathbf{E}_i[\mathrm{err}(P,\mathcal{D}_i)]=\mathrm{err}(P,\widetilde{\mathcal{D}})\leq\epsilon.
	\]

	Let $Y=Y_1Y_2\cdots Y_q$ be distributed according to $\widetilde{\mathcal{D}}_B$, where $\widetilde{\mathcal{D}}_B$ is the marginal distribution of $\widetilde{\mathcal{D}}$ on Bob's inputs. We have
	\begin{align*}
		\mathrm{icost}(P,\widetilde{\mathcal{D}})	&=\mathrm{I}(Y:\mathrm{msg}(P,Y))\\
		&\geq\sum_{i\in[q]}\mathrm{I}(Y_i:\mathrm{msg}(P,Y)),	&	\text{($Y_i$'s are independent)}\\
		&=q\cdot\mathbf{E}_i[\mathrm{I}(Y_i:\mathrm{msg}(P,Y))].
	\end{align*}
	Note that $\mathrm{icost}(P,\widetilde{\mathcal{D}})\leq b_0$ since Bob's first message is of length $b_0$. Hence
	\[
		\mathbf{E}_i[\mathrm{I}(Y_i:\mathrm{msg}(P,Y))]\leq\frac{b_0}{q}.
	\]
	Due to the linearity of expectation,
	\begin{align*}
		&\mathbf{E}_i[\mathrm{err}(P,\mathcal{D}_i)+\mathrm{I}(Y_i:\mathrm{msg}(P,Y))]\\
		=&\mathbf{E}_i[\mathrm{err}(P,\mathcal{D}_i)]+\mathbf{E}_i[\mathrm{I}(Y_i:\mathrm{msg}(P,Y))]\\
		\leq&\epsilon+b_0/q.
		\end{align*}
	By the averaging principle and the concavity of the square root function, there is an $i\in[q]$ such that
	\[
			\mathrm{err}(P,\mathcal{D}_i)+\sqrt{\mathrm{I}(Y_i:\mathrm{msg}(P,Y))}\leq\epsilon+\sqrt{b_0/q}.
	\]
	Fix the $i$ as above. We can now define a private-coin protocol $Q'$ for $\mathsf{LPM}^\Sigma_{m-1,n/q}$ which uses $P$ as a black box. It works as follows: given an input $(x,y)\in S\times S^{n/q}$, Alice constructs a string $\tilde{x}\triangleq s_ix$ and Bob constructs the set of strings $\tilde{y}\triangleq s_1y_1\cup\cdots\cup s_{i-1}y_{i-1}\cup s_iy\cup s_{i+1}y_{i+1}\cup\cdots\cup s_qy_q$ where the $y_j$'s are random sets of strings drawn independently from $\mathcal{D}_B$ using his private-coins. They then run protocol $P$ on input $(\tilde{x},\tilde{y})$ and output the second block of the output of $P$. Note that the $(\tilde{x},\tilde{y})$ is distributed according to $\mathcal{D}_i$ if $(x,y)$ is distributed according to $\mathcal{D}$. Clearly, due to the definition of $\mathsf{LPM}$, $Q'$ works as $P$ works. Therefore,
	\[
		\mathrm{err}(Q',\mathcal{D})\leq\mathrm{err}(P,\mathcal{D}_i).
	\]
	Moreover,
	\[
		\mathrm{icost}(Q',\mathcal{D})=\mathrm{I}(Y_i:\mathrm{msg}(P,Y)).
	\]
	Applying the uninformative message lemma (Lemma~\ref{lem:uniformative message}) to $Q'$, we have a deterministic $\langle\mathbf{A,B},2k\rangle^A$-protocol $Q$ for $\mathsf{LPM}^\Sigma_{m-1,n/q}$ with error at most $\epsilon+\sqrt{b_0/q}$ on distribution $\mathcal{D}$. Applying this to every joint distribution $\mathcal{D}$ over the inputs with Yao's min-max lemma, we prove the second proposition supporting the round elimination lemma.

\begin{suppress}
\todo{editing.., }

Let $\mathcal{A}(m,n,\langle\mathbf{A},\mathbf{B},k\rangle^A,\epsilon)$ denote the statement ``$\LPM$ has a private-coin randomized $\epsilon$-error $\langle\mathbf{A},\mathbf{B},k\rangle^A$-protocol.'' Let $\mathcal{B}(m,n,\langle\mathbf{A},\mathbf{B},k\rangle^B,\epsilon)$ denote the statement ``$\LPM$ has a private-coin randomized $\epsilon$-error $\langle\mathbf{A},\mathbf{B},k\rangle^B$-protocol.'' We know prepared to describe the key lemma of the whole proof.

Be similar with~\cite{CharkReg}, the strategy in this work is to prove a round elimination lemma for $\LPM$ and give a communication lower bound. To this end, we involve a toolkit of three lemmas. The first one is the Message Switching Lemma which says that the first two messages can be switched at the cost of blowing up the sizes of messages such that the first message is eliminated. The sceond one is the Message Compression Lemma which says that the first message can be compressed. We use it to compress the message wich is extremely blowed up in the first lemma. The third one is the Uniformative Message Lemma which says that we can modify the protocol such that the first message is not sent. The last two lemmas increase the error of the protocol by a small amount.

\begin{lemma}[Message Switching Lemma, refined from Lemma 3.1 in~\cite{CharkReg}]\label{lem:message switching}
	Let $P$ be a determinisitic $\langle\mathbf{A},\mathbf{B},k\rangle^A$-protocol with $k\geq2$. Then there exists a determinisitic $\langle(1+\mathbf{A}_1/\mathbf{A}_2)\mathrm{sfx}(\mathbf{A},2),(2^{\mathbf{A}_1}\mathbf{B}_1)\cdot\mathrm{sfx}(\mathbf{B},2),k-1\rangle^B$-protocol that computes the exact same problem as $P$.
\end{lemma}
The generalized lemma is refined from Lemma 3.1 in~\cite{CharkReg}. The proof of the lemma is omitted since it will be same with the one of Lemma 3.1 in~\cite{CharkReg}.
\par
To show other two lemmas in the toolkit, we need some notations and a definition. Let $P$ be a communication protocol and $\mathcal{D}$ a distribution on the possible inputs to $P$. Note that $\mathcal{D}$ is a joint distribution for $X$, which is Alice's input, and $Y$, which is Bob's input. We let $\mathcal{D}_A$ denote the distribution of $X$, i.e. the marginal distribution of $X$ of $\mathcal{D}$. Simiarly we let $\mathcal{D}_B$ denote the distribution of $Y$. We let $\mathrm{err}(P,\mathcal{D})$ denote the distributional error probability of $P$ under the input distribution $\mathcal{D}$. We use $\mathrm{msg}(P,x)$ to denote the first message in protocol $P$ if the sender's input is $x$. Note that $\mathrm{msg}(P,x)$ will be a random variable if $P$ is a randomized protocol.

\begin{definition}[Information Cost, Definition 3.2 in ~\cite{CharkReg}] The information cost of a private-coin protocol $P$ with respect to input distribution $\mathcal{D}$, denoted icost$(P,\mathcal{D})$, is defined to be the mutual information $I(X:\mathrm{msg}(P,X))$, where $X$ is a random input drawn from $\mathcal{D}_A$ (if Alice starts $P$) or $\mathcal{D}_B$ (if Bob starts $P$). Note that the notation only deals with the first message of a protocol.
\end{definition}
\begin{lemma}[Message Compression Lemma\cite{CharkReg}] \label{lem:message compression}
	Let $P$ be a private-coin $\langle\mathbf{A},\mathbf{B},k\rangle^A$-protocol for a communication problem $\rho$. Then for any input distribution $\mathcal{D}$ and any $a>0$, there is a determinisitic $\langle a\cdot\mathrm{sfx}(\mathbf{A},2),\mathbf{B},k\rangle^A$-protocol $P'$ for $\rho$ such that $\mathrm{err}(P',\mathcal{D})\leq\mathrm{err}(P,\mathcal{D})+(2\cdot\mathrm{icost}(P,\mathcal{D})+C)/a$, where $C>0$ is a universal constant.
\end{lemma}
\begin{lemma}[Uniformative Message Lemma\cite{senp03}, Lemma 3.3 in~\cite{CharkReg}]\label{lem:uniformative message}
	Let $P$ be a private-coin $\langle\mathbf{A},\mathbf{B},k\rangle^A$-protocol for a communication problem $\rho$. Then for any input distribution $\mathcal{D}$, there is a determinisitic $\langle\mathrm{sfx}(\mathbf{A},2),\mathbf{B},k-1\rangle^B$-protocol $P'$ for $\rho$ such that $\mathrm{err}(P',\mathcal{D})\leq\mathrm{err}(P,\mathcal{D})+\sqrt{\mathrm{icost}(P,\mathcal{D})}$.
\end{lemma}
		Let $\mathcal{A}(m,n,\langle\mathbf{A},\mathbf{B},k\rangle^A,\epsilon)$ denote the statement ``$\LPM$ has a private-coin randomized $\epsilon$-error $\langle\mathbf{A},\mathbf{B},k\rangle^A$-protocol.'' Let $\mathcal{B}(m,n,\langle\mathbf{A},\mathbf{B},k\rangle^B,\epsilon)$ denote the statement ``$\LPM$ has a private-coin randomized $\epsilon$-error $\langle\mathbf{A},\mathbf{B},k\rangle^B$-protocol.'' We know prepared to describe the key lemma of the whole proof.
\begin{lemma}[Round Elimination Lemma for $\LPM$, refined from Lemma 4.2 in~\cite{CharkReg}]
	Let $m,n,p,q,b_0$ be positive integer such that $p|m$, $q|n$. Let $\epsilon,\delta$ be positive reals, $\mathbf{A,B}$ be positive vectors, and $C$ be the universal constant from the Message Compression Lemma.
	\begin{enumerate}
		\item If $k\geq2$ and $2\mathbf{A}_1/p\geq C$, then $\mathcal{A}(m,n,\langle\mathbf{A},\mathbf{B},k\rangle^A,\epsilon)\implies\mathcal{B}(m/p,n,\langle(1+\frac{2\mathbf{A}_1}{\delta p\mathbf{A}_2})\mathrm{sfx}(\mathbf{A},2),2^{2\mathbf{A}_1/(\delta p)}\mathbf{B}_1\cdot\mathrm{sfx}(\mathbf{B},2),k-1\rangle^B,\epsilon+2\delta)$.
		\item If $n\leq|\Sigma|$, then $\mathcal{B}(m,n,\langle\mathbf{A},b_0\cdot\mathbf{B},k\rangle^B,\epsilon)\implies\mathcal{A}(m-1,n/q,\langle\mathbf{A},\mathbf{B},k-1\rangle^A,\epsilon+\sqrt{b_0/q})$.
	\end{enumerate}
\end{lemma}
We omit the proof of the lemma since it will be same with the lemma 4.2 in~\cite{CharkReg} expect we prove it in a generalized background. Note that the generalized lemma implies the lemma in~\cite{CharkReg} but the one in~\cite{CharkReg} does not. Combining the two parts of above round elimination lemma, and weakening the resulting statement from $m/p-1$ to $m/2p$, we have the following corollary.

\begin{corollary}[improved from Corrollary 4.3 in~\cite{CharkReg}]\label{cor:key result}
	Let $m,n,p,q$ be positive integer such that $(2p)|m$, $q|n$, $n\leq|\Sigma|$. Let $\epsilon,\delta$ be positive reals, $\mathbf{A,B}$ be positve vectors, and $C$ be the universal constant from the Message Compression Lemma. If $k\geq2$ and $2\mathbf{A}_1/p\geq C$, then
	\[
		\mathcal{A}(m,n,\langle\mathbf{A,B},k\rangle^A,\epsilon)\implies\mathcal{A}\Bigg(\frac{m}{2p},\frac{n}{q},\langle(1+\frac{2\mathbf{A}_1}{\delta p\mathbf{A}_2})\mathrm{sfx}(\mathbf{A},2),\mathrm{sfx}(\mathbf{B},2),k-2\rangle^A,\epsilon+2\delta+\sqrt{\frac{2^{2\mathbf{A_1}/(\delta p)}\mathbf{B}_1}{q}}\Bigg).
	\]
\end{corollary}
\end{suppress}

}

\subsection{Proof of the lower bound}
We now prove the communication lower bound for $\LPM$, by the round elimination tool we setup in previous sections.

\begin{theorem}[communication lower bound for $\mathsf{LPM}$]
\label{thm:main lower bound}
For any $c_1,c_2>0$, there exists a $c_3>0$ such that the followings hold. Let $n,d\geq1$ be sufficiently large integers, and suppose that $d\le 2^{\sqrt{\log n}}$ and $n\le 2^{d^{0.99}}$.
Let $\gamma\geq3$ and $m=\lfloor(\log d)^{\eta\beta}\rfloor$, where $\eta$ and $\beta$ are as defined in~\eqref{eq:eta-beta}.
Let $\Sigma$ be a set of cardinality $\lceil 2^{d^{0.99}}\rceil$.
Let $1\le k\le \frac{\log\log d}{2\log\log\log d}$ be an integer.
Let $\mathbf{A},\mathbf{B}\in\mathbb{R}_{> 0}^k$ be in the form that $a_i=c_1t_i\log n$ and $b_i=t_id^{c_2}$ for some $t_i>0$ for every $1\le i \le k$. If $\LPM$ has a private-coin $\langle\mathbf{A,B},2k\rangle^A$-protocol, then $t\triangleq\sum_{i=1}^kt_i>\frac{c_3}{k}(\log_\gamma d)^{1/k}$.
\end{theorem}
\begin{proof}
	Although \ifabs{$\mathbf{A}=(a_1,\ldots,a_k)$ and $\mathbf{B}=(b_1,\ldots,b_k)$,}{$\mathbf{A}=(a_1,a_2,\ldots,a_k)$ and $\mathbf{B}=(b_1,b_2,\ldots,b_k)$,} we additionally define
	\begin{align}\label{eq:k+1}
		a_{k+1}\triangleq a_1,\quad  b_{k+1}\triangleq b_1, \quad\text{and }t_{k+1}\triangleq t_1.
	\end{align}	
	We set $\delta=\frac{1}{4k}$, $p=\frac{1}{2}m^{1/k}$, and further define
	\[
		c_5\triangleq\max_{1\le i\le k}\frac{\log\left({b_i}/{\delta^2}\right)}{\log n}\frac{m^{1/k}}{kt_{i+1}}.
	\]
	By  definitions of $\beta$ and $\eta$ in~\eqref{eq:eta-beta}, we have
	\[
		\eta\beta=1-\frac{c_4+\log\log\gamma}{\log\log d}+\frac{c_4\log\log\gamma}{(\log\log d)^2},
	\]
	and
	\begin{align*}
		m=(\log d)^{\eta\beta}
		=\frac{(\log d)2^{c_4\log\log\gamma/\log\log d}}{2^{c_4}\log\gamma}
		=\Theta(\log_\gamma d).
	\end{align*}
	Therefore, assuming that $d\le2^{\sqrt{\log n}}$ and $\gamma\ge 3$, we have
	\begin{align*}
		c_5
		&=\max_{i\in[k]}O\left(\frac{(\log d)^{1/k}(\log k + \log d)}{kt_{i+1}\log n}\right)\\
		&=O\left(\frac{(\log d)^2}{\log n}\right)=O(1),
	\end{align*}
where the constant factor depends on $c_2$.

Furthermore, since $k\leq\frac{\log\log d}{2\log\log\log d}$, 
it can be verified that
\begin{align}
		m>k^{2k}\label{ieq:m k}.
\end{align}
	Now we define
	\[
		\xi\triangleq\frac{m^{1/k}}{k}=\Theta\left(\frac{(\log_\gamma d)^{1/k}}{k}\right).
	\]
We start our proof by assuming $\LPM$ has a private-coin $\langle \mathbf{A},\mathbf{B}, 2k\rangle^A$-protocol with error probability $1/8$ and
	\[
		t=c_3\xi=\xi/(c_5+16c_1e^{16}),
	\]
and derive an impossible result, which will prove that $t>c_3\xi$. For notational convenience, we ignore divisibility issues.

With the above assumption, we make the following claim.
\begin{claim}\label{claim-induction}
For any non-negative integer $i\leq k$, $\mathsf{LPM}^{\Sigma}_{m_i,n_i}$ has a private-coin $\langle \mathbf{A}_i,\mathbf{B}_i, 2(k-i)\rangle^A$-protocol with error probability $\frac{1}{8}+3i\delta$, where
\begin{align*}
m_i &=\frac{ma_{i+1}}{(2p)^ia_1},\\
n_i &=n^{1-\frac{1}{t}\sum_{j\le i}t_{j+1}},\\
\mathbf{A}_i &=\prod_{j=1}^i\left(1+\frac{2a_j}{a_{j+1}\delta p}\right)\mathbf{A}^{(i+1)-},\\
\mathbf{B}_i &= \mathbf{B}^{(i+1)-}.
\end{align*}
\end{claim}
We prove this claim by induction on $i$. For $i=0$, the claim holds by our assumption.
For  induction hypothesis: assume the claim for an $i<k$. We then prove the claim for $i+1$.

We choose $p_{i+1}=\frac{a_{i+1}}{a_{i+2}}p$ and $q_{i+1}=n^{t_{i+2}/t}$.
We claim that
\begin{align}
2p_{i+1}\le m_i=\frac{ma_{i+1}}{(2p)^i a_1}.\label{eq:p-bound}
\end{align}
When $i+1=k$, this is obviously true, because $\frac{2p_{k}}{m_{k-1}}=\frac{2a_{k}p}{a_{k+1}}\frac{a_1(2p)^{k-1}}{m a_{k}}=1$ since $a_{k+1}=a_1$ and $2p=m^{1/k}$, thus $2p_k=m_{k-1}$; and when $i< k-1$, we have $\frac{2p_{i+1}}{m_{i}}=\frac{a_1(2p)^i}{m a_{i+2}}\le \frac{a_1}{a_{i+2}m^{1/k}}$ because $2p=m^{1/k}$, and $\frac{a_1}{a_{i+2}}\le t_i\le t\le m^{1/k}$ (or otherwise $t>m^{1/k}=\Omega((\log_\gamma d)^{1/k})$ and there is nothing to prove). Therefore, $2p_{i+1}\le m_i$ holds for all $i$.
	
It is also  obvious that 
\begin{align}
q_{i+1}\leq n^{1-\frac{1}{t}\sum_{j=1}^it_{j+1}}\leq|\Sigma|.\label{eq:q-bound}
\end{align}

On the other hand, the quantity $\frac{2a_{i,1}}{p_{i+1}}$, where \ifabs{$a_{i,1}=a_{i+1}\prod_{j=1}^i$\\$(1+\frac{2 a_j}{a_{j+1}\delta p})$}{$a_{i,1}=a_{i+1}\prod_{j=1}^i(1+\frac{2 a_j}{a_{j+1}\delta p})$} is the first entry of $\mathbf{A}_i$, is bounded as below:
\begin{align*}
\frac{2 a_{i+1}}{p_{i+1}}\prod_{j=1}^i\left(1+\frac{2 a_j}{a_{j+1}\delta p}\right)	&\geq\frac{2a_{i+2}}{p}		
\ifabs{}{\geq\frac{2c_1t_{i+2}\log n}{\frac{1}{2}m^{1/k}}}
>\frac{2c_1\log n}{(\log d)^{1/k}},
	\end{align*}
which is $\omega(1)$ for $d\le2^{\sqrt{\log n}}$. Therefore, it holds that $\frac{2a_{i,1}}{p_{i+1}}>C$ for the first entry $a_{i,1}$ of $\mathbf{A}_i$, where $C$ is the universal constant in Lemma~\ref{lemma-round-elimination}. Together with~\eqref{eq:p-bound} and~\eqref{eq:q-bound}, the condition of the round elimination lemma (Lemma~\ref{lemma-round-elimination}) is satisfied. We now apply the round elimination lemma to the protocol assumed by the induction hypothesis, to obtain a private-coin $\langle \mathbf{A}_{i+1},\mathbf{B}_{i+1}, 2(k-i-1)\rangle^A$-protocol for $\mathsf{LPM}^{\Sigma}_{m_{i+1},n_{i+1}}$ with error probability $\frac{1}{8}+3\delta+2\delta+\delta'$, where
\[
\delta'=\sqrt{\frac{b_{i+1}}{q_{i+1}}}\exp\left(\frac{a_{i+1}\ln 2}{\delta p_{i+1}}\prod_{j=1}^i\left(1+\frac{2a_j}{a_{j+1}\delta p}\right)\right).
\]
We then show that $\delta'\le\delta$. Note that this will finish our induction and prove Claim~\ref{claim-induction}.

By~(\ref{ieq:m k}), we have $\delta p/2=\frac{1}{16}m^{1/k}/k>\frac{k}{16}$, thus
\ifabs{
\begin{align*}
		\prod_{j=1}^i\left(1+\frac{2a_j}{a_{j+1}\delta p}\right)	
		&\leq\prod_{j=1}^k\left(1+\frac{2a_j}{a_{j+1}\delta p}\right)\\
		&\leq \left(1+\frac{2}{\delta p}\right)^k\\
		&\leq \mathrm{e}^{16}.
	\end{align*}
	Where the second inequality is by the Lagrange multipliers. }
{
\begin{align*}
		\prod_{j=1}^i\left(1+\frac{2a_j}{a_{j+1}\delta p}\right)	
		&\leq\prod_{j=1}^k\left(1+\frac{2a_j}{a_{j+1}\delta p}\right)\\
		&\leq \left(1+\frac{2}{\delta p}\right)^k	&	(\text{Lagrange multipliers})\\
		&\leq \mathrm{e}^{16}.
	\end{align*}
}Therefore,
	\begin{align*}
\delta'	\leq\sqrt{\frac{b_{i+1}}{q_{i+1}}}2^{\frac{e^{16}a_{i+1}}{\delta p_{i+1}}}
		=\sqrt{b_{i+1}}n^{8e^{16}c_1t_{i+2}k/m^{1/k}-t_{i+2}/2t}.
	\end{align*}
Recall that
		$t=c_3\xi=\xi/(c_5+16c_1e^{16})$ and
		\ifabs{
		$c_5=\max_{1\le i\le k}$\\$\frac{\log\left({b_i}/{\delta^2}\right)m^{1/k}}{kt_{i+1}\log n}$.
	}
	{
		$c_5=\max_{1\le i\le k}\frac{\log\left({b_i}/{\delta^2}\right)}{\log n}\frac{m^{1/k}}{kt_{i+1}}$.
	}
We have
	\begin{align*}
				&c_5\geq \frac{m^{1/k}\log(b_{i+1}/\delta^2)}{kt_{i+2}\log(n)}\\
				\iff&16c_1e^{16}-\frac{m^{1/k}\log(\delta^2/b_{i+1})}{kt_{i+2}\log(n)}\leq16c_1e^{16}+c_5=\frac{1}{c_3}\\
			\iff&\left(16c_1e^{16}-\frac{1}{c_3}\right)\frac{t_{i+2}k}{m^{1/k}}\leq\frac{\log(\delta^2/b_{i+1})}{\log n}\\
			\iff&\frac{16e^{16}c_1t_{i+2}k}{m^{1/k}}-\frac{t_{i+2}}{t}\leq\frac{\log(\delta^2/b_{i+1})}{\log n}\\
			\iff&\delta'^2\leq n^{16e^{16}c_1t_{i+2}k/m^{1/k}-t_{i+2}/t}b_{i+1}\leq\delta^2.
	\end{align*}
And we prove that $\delta'\le \delta$. Claim~\ref{claim-induction} is proved.
	

Now let $i=k$ in Claim~\ref{claim-induction}. We have a private-coin protocol without message exchange between Alice and Bob but solving $\mathsf{LPM}^{\Sigma}_{1,1}$ with error probability at most $\frac{1}{8}+\frac{3}{4}$. However, this is impossible due to the following claim.

\begin{claim}\label{claim:impossible}
Any private-coin protocol for $\mathsf{LPM}^\Sigma_{1,1}$ without message exchange can succeed with probability at most $1/|\Sigma|$ in the worst case.
\end{claim}
By Yao's min-max principle, it is sufficient to prove the lower bound for deterministic protocols on a uniform random inputs.
Note that the only thing a deterministic Alice can do without communication is to pick a string in $\Sigma$ and output it, but this can only succeed with probability at most $1/|\Sigma|$ for a random input.\end{proof}


Theorem~\ref{thm:main lower bound} together with the translation from cell-probing scheme to communication protocol (Proposition~\ref{thm:cell-probe to communication}) and the reduction from $\LPM$ to $\ANNS$ (Lemma~\ref{lem:reduction}) prove the $k$-round cell-probe lower bound for $\ANNS$ (Theorem~\ref{theorem-ANNS-lower-bound}), the main result of this section.

\bibliographystyle{abbrv}
\bibliography{paper}

\end{document}